\documentclass[11pt]{article}

\usepackage{xspace}
\usepackage{verbatim}
\usepackage{color}
\usepackage{graphicx}
\usepackage{tabularx}

\usepackage{amsmath}
\usepackage[showonlyrefs]{mathtools}
\usepackage{mathstyle}
\usepackage{empheq}
\usepackage{amstext,amssymb,amsfonts}
\usepackage{fullpage}
\usepackage{nicefrac}
\usepackage{bm}

\usepackage{todonotes}

\newcounter{mynotes}
\usepackage{textcomp,setspace}

\usepackage{nameref}
\usepackage[pagebackref,colorlinks,linkcolor=blue,filecolor = blue,citecolor = blue, urlcolor = blue, hyperfootnotes=false]{hyperref}
\usepackage{color}
\usepackage[capitalize, noabbrev]{cleveref}

\usepackage{amsthm}
\usepackage{thmtools}

\declaretheorem[within=section]{theorem}

\declaretheorem[sibling=theorem]{lemma}

\declaretheorem[sibling=theorem]{definition}

\declaretheorem[sibling=theorem]{remark}

\newcounter{termcounter}
\renewcommand{\thetermcounter}{\Alph{termcounter}}
\crefname{term}{term}{terms}
\creflabelformat{term}{\textup{(#2#1#3)}}

\makeatletter
\def\term{\@ifnextchar[\term@optarg\term@noarg}
\def\term@optarg[#1]#2{%
  \textup{(#1)}%
  \def\@currentlabel{#1}%
  \def\cref@currentlabel{[][2147483647][]#1}%
  \cref@label[term]{#2}}
\def\term@noarg#1{%
  \refstepcounter{termcounter}%
  \textup{(\thetermcounter)}%
  \cref@label[term]{#1}}
\makeatother

\newcommand{\mv}[1]{\mathbf {#1}}
\newcommand{\msf}[1]{\mathsf {#1}}
\newcommand{\mcl}[1]{\mathcal {#1}}

\newcommand{\mt}[1]{\text{#1}}

\newcommand{\ignore}[1]{}

\newcommand{\defeq}{\stackrel{\mathrm{def}}=}

\newcommand{\paren}[1]{(#1 )}
\newcommand{\Paren}[1]{\left(#1 \right )}

\definecolor{DSred}{rgb}{1,0,0}

\renewcommand{\leq}{\leqslant}
\renewcommand{\geq}{\geqslant}
\renewcommand{\ge}{\geqslant}
\renewcommand{\le}{\leqslant}
\renewcommand{\epsilon}{\varepsilon}
\newcommand{\eps}{\epsilon}

\newcommand{\Psymb}{{\bf Pr}}

\DeclareMathOperator*{\ProbOp}{\Psymb}

\renewcommand{\Pr}{\ProbOp}

\usepackage{enumitem}
\usepackage{caption}
\usepackage{algorithm2e,caption}

\usepackage{enumitem}
\usepackage{mdwlist}

\DeclareMathOperator*{\argmin}{arg\,min}
\DeclareMathOperator*{\Pa}{\bm{\mathsf{Pa}}}
\DeclareMathOperator*{\pa}{\bm{\mathsf{pa}}}
\DeclareMathOperator*{\ePa}{\bm{\mathsf{Pa^{+}}}}

\DeclareMathOperator*{\An}{\bm{\mathsf{An}}}
\DeclareMathOperator*{\an}{\bm{\mathsf{an}}}
\DeclareMathOperator*{\De}{\bm{\mathsf{De}}}
\DeclareMathOperator*{\de}{\bm{\mathsf{de}}}
\DeclareMathOperator*{\COA}{\mathsf{COA}}

\DeclareMathOperator*{\Neighbor}{\bm{\mathsf{Ne}}}
\DeclareMathOperator*{\neighbor}{\bm{\mathsf{ne}}}
\newcommand{\cit}{\mathsf{C2ST}(G,\epsilon)} 
\newcommand{\cgft}{\mathsf{CGFT}(G, \mathcal{M}, \epsilon)} 
\newcommand{\clp}{\mathsf{CL}(G, \epsilon)} 

\newcommand{\smbn}{\textsc{SMBN}} 
\newcommand{\smcg}{\textsc{SMCG}} 
\newcommand{\cbn}{\textsc{CBN}} 
\newcommand{\bn}{\textsc{BN}} 

\newcommand{\indep}{\rotatebox[origin=c]{90}{$\models$}}

\newcommand{\Opt}{\msf{Opt}}

\allowdisplaybreaks[1]

\title{Learning and Testing Causal Models with Interventions}
\usepackage{times}

\author{
\begin{tabular}{c c c c}
Jayadev Acharya\thanks{Email: \texttt{acharya@cornell.edu}. Supported by a Cornell University Startup.} & Arnab Bhattacharyya\thanks{ Email: \texttt{arnabb@iisc.ac.in}. Partially supported by DST Ramanujan grant DSTO1358 and DRDO Frontiers Project DRDO0687.} & Constantinos Daskalakis\thanks{Email: \texttt{costis@csail.mit.edu}. Partially supported by NSF CCF-1650733, CCF-1617730, CCF-1741137.} & Saravanan Kandasamy\thanks{Email: \texttt{saravan.tuty@gmail.com}. Partially supported by DRDO Frontiers Project DRDO0687.}\\[1em]
School of ECE & CSA Department & CSAIL & CSA Department\\
Cornell University & Indian Institute of Science & MIT & Indian Institute of Science
\end{tabular}
}

\author {
Jayadev Acharya\thanks{Supported by a Cornell University startup grant.} \\
School of ECE\\
 Cornell University\\
\texttt{acharya@cornell.edu}
\and
Arnab Bhattacharyya\thanks{Partially supported by DST Ramanujan grant DSTO1358 and DRDO Frontiers Project DRDO0687.} \\
CSA Department\\
Indian Institute of Science \\
\texttt{arnabb@iisc.ac.in}
\and
Constantinos Daskalakis\thanks{Partially supported by NSF CCF-1650733, CCF-1617730, CCF-1741137.}  \\
EECS\\
 MIT\\
\texttt{costis@csail.mit.edu}
\and
Saravanan Kandasamy\thanks{Partially supported by DRDO Frontiers Project DRDO0687.} \\
CSA Department\\
Indian Institute of Science \\
\texttt{saravan.tuty@gmail.com}
}

\ignore{
\thanks{School of ECE, Cornell University. Email: \texttt{acharya@cornell.edu}. Supported by a Cornell University Startup. }\qquad \qquad \qquad
&Arnab Bhattacharyya\thanks{Dept.~Computer Science \& Automation, Indian Institute of Science. Email: \texttt{arnabb@iisc.ac.in}. Partially supported by DST Ramanujan grant DSTO1358 and DRDO Frontiers Project DRDO0687.}\qquad \qquad \qquad
& Constantinos Daskalakis\thanks{CSAIL, Massachusetts Institute of Technology. Email: \texttt{costis@csail.mit.edu}. Partially supported by NSF CCF-1650733, CCF-1617730, CCF-1741137.} \\[1.5em]
&Saravanan Kandasamy\thanks{Dept.~Computer Science \& Automation, Indian Institute of Science. Email: \texttt{saravan.tuty@gmail.com}. Partially supported by DRDO Frontiers Project DRDO0687.} &
\end{tabular}
}

\ignore{
\author{
Jayadev Acharya\thanks{Supported by a Cornell University Startup.}\\
School of ECE\\
Cornell University\\
\texttt{acharya@cornell.edu}

\and

Arnab Bhattacharyya\thanks{Partially supported by DST Ramanujan grant DSTO1358 and DRDO Frontiers Project DRDO0687.}\\
Computer Science \& Automation Dept.\\
Indian Institute of Science\\
 \texttt{arnabb@iisc.ac.in}

\and
Constantinos Daskalakis\thanks{Partially supported by NSF CCF-1650733, CCF-1617730, CCF-1741137.}\\
CSAIL\\
Massachusetts Institute of Technology\\
\texttt{costis@csail.mit.edu}
}
}

\begin{document}

\maketitle

\begin{abstract}
We consider testing and learning problems on causal Bayesian networks as defined by Pearl~\cite{Pearl00}. Given a causal Bayesian network $\mathcal{M}$ on a graph with $n$ discrete variables and bounded in-degree and bounded ``confounded components'', we show that $O(\log n)$ interventions on an unknown causal Bayesian network $\mathcal{X}$ on the same graph, and $\tilde{O}(n/\epsilon^2)$ samples per intervention, suffice to efficiently distinguish whether $\mathcal{X}=\mathcal{M}$ or whether there exists some intervention under which $\mathcal{X}$ and $\mathcal{M}$ are farther than $\epsilon$ in total variation distance.  We also obtain sample/time/intervention efficient algorithms for: (i) testing the identity of two unknown causal Bayesian networks on the same graph; and (ii) learning a causal Bayesian network on a given graph.  Although our algorithms are non-adaptive, we show that adaptivity does not help in general: $\Omega(\log n)$ interventions are necessary for testing the identity of two unknown causal Bayesian networks on the same graph, even adaptively.  Our algorithms are enabled by a new subadditivity inequality for the squared Hellinger distance between two causal Bayesian networks.
\end{abstract}

\section{Introduction}

A central task in statistical inference is learning properties of a high-dimensional distribution over some variables of interest given observational data. However, probability distributions only capture the association between  variables of interest and may not  suffice to predict what the consequences would be of setting some of the variables to particular values. A standard example illustrating the point is this: From observational data, we may learn that atmospheric air pressure and the readout of a barometer are correlated. But can we  predict whether the atmospheric pressure would stay the same or go up if the barometer readout was forcefully increased by moving its needle? 

Such issues are at the heart of {\em causal inference}, where the goal is to learn a {\em causal model} over some variables of interest, which can predict the result of {\em external interventions} on the variables. For example, a causal model on two variables of interest $X$ and $Y$ need not only determine conditional probabilities of the form $\Pr[Y \mid X = x]$, but also {\em interventional probabilities} $\Pr[Y \mid do(X=x)]$ where, following Pearl's notation \cite{Pearl00}, $do(X=x)$ means that $X$ has been forced to take the value $x$ by an external action. In our previous example, $\Pr[\text{Pressure} \mid do(\text{Barometer} = b)] = \Pr[\text{Pressure}]$ but  $\Pr[\text{Barometer} \mid do(\text{Pressure} = p)] \neq \Pr[\text{Barometer}]$, reflecting that the atmospheric pressure causes the barometer readout, not the other way around. 

Causality has been the focus of extensive study, with a wide range of analytical frameworks proposed to capture causal relationships and perform causal inference. A prevalent class of causal models are {\em  graphical causal models}, going back to Wright \cite{Wright21} who introduced such models for path analysis, and Haavelmo \cite{Haavelmo43} who used them to define structural equation models. Today, graphical causal models are widely used to represent causal relationships in a variety of ways~\cite{SDLC93, GC99, Pearl00, SGS00, Neapolitan03, KF09}.  

In our work, we focus on the central model of {\em causal Bayesian networks}  (CBNs) \cite{Pearl00, SGS00, Neapolitan03}. Recall that a (standard) Bayesian network is a distribution over several random variables that is associated with a directed acyclic graph. The vertices of the graph are the random variables over which the distribution is defined, and the graph describes conditional independence properties of the distribution. In particular, every variable is independent of its non-descendants, conditioned on the values of its parents in the graph. A CBN is also associated with a directed acyclic graph (DAG) whose vertices are the random variables on which the distribution is defined. However, a CBN is not a single distribution over these variables  but the collection of all possible interventional distributions, defined by setting any subset of the variables to any set of values. In particular, every vertex is both a variable $V$ and a {\em mechanism} to generate the value of $V$ given the values of the parent vertices, and the interventional distributions are defined in terms of these mechanisms.

\ignore{
To illustrate the difference between Bayesian networks and causal Bayesian networks, observe that a (standard) Bayesian network associated with the graph $X \rightarrow Y \rightarrow Z$, can also be associated with the graphs $X \leftarrow Y \leftarrow Z$ and $X \leftarrow Y \rightarrow Z$, as all graphs declare that $X$ and $Z$ are independent conditioned on $Y$. On the other hand, a causal Bayesian network associated with one of these graphs, might not be possible to associate to any of the other graphs, as the three graphs imply different causal relationships. Back to our original example, take $X=\text{Barometer Readout}$,  $Y =\text{Atmospheric Pressure}$, $Z=\text{Altitude}$. Clearly, a causal Bayesian network that is able to generate all interventional distributions on variables $X, Y, Z$ should better use the graph $X \leftarrow Y \leftarrow Z$, as only this one captures the causal relationship among these variables.
}
We allow CBNs to contain both observable and unobservable (hidden) random variables. Importantly, we allow {\em unobservable confounding variables}. These are variables that are not observable, yet they are ancestors of at least two observable variables. These are especially tricky in statistical inference, as they may lead to spurious associations. 

\subsection{Our Contributions}
Consider the following situations:
\begin{enumerate}
\item
An engineer designs a large circuit using a circuit simulation program and then builds it in hardware. The simulator predicts relationships between the voltages and currents at different nodes of the circuit. Now, the engineer would like to verify whether the simulator's predictions hold for the real circuit by doing a limited number of experiments (e.g., holding some voltages at set levels, cutting some wires, etc.). If not, then she would want to learn a model for the system that has sufficiently good accuracy.
\item
A biologist is studying the role of a set of genes in migraine. He would like to know whether the mechanisms relating the products of these genes are approximately the same for patients with and without migraine.
He has access to tools (e.g., CRISPR-based gene editing technologies \cite{dixit16}) that generate data for gene activation and knockout experiments. 
\end{enumerate}
Motivated by such scenarios, we study the problems of hypothesis testing and learning CBNs when both observational and interventional data are available.  The main highlight of our work is that we prove bounds on the number of samples, interventions, and time steps required by our algorithms.

\ignore{This setting has been a recent focus of study \cite{hauser12,wang17,yang18}, motivated by advances in genomics which allow high-resolution observational and interventional data for gene expression using flow cytometry and CRISPR technologies \cite{sachs05,macosko15,dixit16}. Our algorithms are non-parametric, which is also crucial in genomics studies where the data is known to be intrinsically non-Gaussian.}
\ignore{
From observational and interventional data, discovering `causal relations' or estimation of `interventional markov equivalent classes' has been studied recently \cite{yang18,wang17,hauser12}.  This setting is motivated by the recent advancements in biology which allow obtaining high resolution observational and interventional data (of the order of 100,000 cells / samples on 20,000 genes / variables) for gene regulatory networks \cite{sachs05,macosko15,dixit16}, where the interventions can be performed using chemical reagents and gene deletions.  In this paper, we study hypothesis testing and learning problems for causal models using both observational and interventional data.  Our algorithms are non-parametric, which is also crucial in gene regulatory networks where the data is known to be intrinsically non-Gaussian.}  

To define our problems precisely, we need to specify what we consider to be a good approximation of a causal model.
Given $\eps \in (0,1)$, we say that two causal models $\mathcal{M}$ and $\mathcal{N}$ on a set of variables $\mv{V} \cup \mv{U}$ (observable and unobservable resp.) are {\em $\eps$-close} (denoted $\Delta(\mathcal{M},\mathcal{N}) \leq \eps$) if for every subset $\mv{S}$ of $\mv{V}$ and assignment $\mv{s}$ to $\mv{S}$, performing the same intervention $do(\mv{S}=\mv{s})$ to both $\mathcal{M}$ and $\mathcal{N}$ leads to the two interventional distributions being $\eps$-close to each other in total variation distance. Otherwise, the two models are said to be {\em $\eps$-far} and $\Delta(\mathcal{M},\mathcal{N}) > \eps$.

Thus, two models $\mathcal{M}$ and $\mathcal{N}$ are close according to the above definition if there is {\em no} intervention which can make the resulting distributions differ significantly. This definition is motivated by the philosophy articulated by Pearl (pp.~414, \cite{Pearl00}) that ``causation is a summary of behavior under intervention''. Intuitively, if there is some intervention that makes $\mathcal{M}$ and $\mathcal{N}$ behave differently, then $\mathcal{M}$ and $\mathcal{N}$ do not describe the same causal process. Without having any prior information about the set of relevant interventions, we adopt a worst-case view and simply require that causal models $\mathcal{M}$ and $\mathcal{N}$ behave similarly for every intervention to be declared close to each other.\footnote{To quote Pearl again, ``It is the
nature of any causal explanation that its utility be proven not over standard situations but
rather over novel settings that require innovative manipulations of the standards.'' (pp. 219, \cite{Pearl00}).}


The {\em goodness-of-fit testing} problem can now be described as follows.
Suppose that a collection  $\mv{V} \cup \mv{U}$ (observable and unobservable resp.) of $n$ random variables  are causally related to each other.
Let $\mathcal{M}$ be a hypothesized causal model for $\mv{V} \cup \mv{U}$ that we are given explicitly. Suppose that the true model to describe the causal relationships is an unknown $\mathcal{X}$. Then, the goodness-of-fit testing problem is to distinguish between: (i) $\mathcal{X} = \mathcal{M}$, versus (ii) $\Delta(\mathcal{X},\mathcal{M}) > \eps$, by sampling from and experimenting on $\mv{V}$, i.e.~forcing some variables in $\mv{V}$ to certain values and sampling from the thus intervened upon distribution.

We study goodness-of-fit testing assuming $\mathcal{X}$ and $\mathcal{M}$ are causal Bayesian networks over a known DAG $G$. Given a DAG $G$, CBN $\mathcal{M}$ and $\eps > 0$, we denote the corresponding goodness-of-fit testing problem $\mathsf{CGFT}(G, \mathcal{M}, \eps)$. For example,  the engineer above, who wants to determine whether the circuit behaves as the simulation software predicts, is interested in the problem $\mathsf{CGFT}(G, \mathcal{M}, \eps)$ where $\mathcal{M}$ is the simulator's prediction, $G$ is determined by the circuit layout, and $\epsilon$ is a user-specified accuracy parameter.
\ignore{

Here is a way to interpret the problem $\mathsf{CGFT}(G, \mathcal{M}, \eps)$. The DAG $G$ reflects a scientist's qualitative understanding of the causal mechanisms behind a system.
She hypothesizes that the model that describes the system is a CBN $\mathcal{M}$ (perhaps, itself the output of another algorithm) over $G$. Now, if an algorithm for $\mathsf{CGFT}(G, \mathcal{M}, \eps)$ accepts, it means that $\mathcal{M}$ is $\eps$-close to the true causal model for the system, meaning that $\mathcal{M}$ can be used to estimate the result of any intervention up to $\eps$ error in total variation distance.}
 Here is our theorem {for goodness-of-fit testing.} 

\begin{theorem}[Goodness-of-fit Testing -- Informal]\label{thm:maininf}
Let $G$ be a DAG on $n$ vertices with bounded in-degree and bounded ``confounded components.'' Let $\mathcal{M}$ be a given CBN over $G$. Then, there exists an algorithm solving $\mathsf{CGFT}(G, \mathcal{M}, \eps)$ that makes $O(\log n)$ interventions, takes $\tilde{O}(n/\eps^2)$ samples per intervention and runs in time $\tilde{O}(n^2/\eps^2)$.  Namely, the algorithm gets access to a CBN $\mathcal{X}$ over $G$, accepts  with probability $\geq 2/3$  if $\mathcal{X} = \mathcal{M}$ and rejects with probability $\geq 2/3$ if $\Delta(\mathcal{X},\mathcal{M}) > \eps$. 
\end{theorem}
By ``confounded component'' in the above statement, we mean a {\em c-component} in $G$, as defined in Definition~\ref{defn:c-component}. Roughly, a c-component is a maximal set of observable vertices that are pairwise connected by paths of the form $V \leftarrow U \rightarrow V \leftarrow U \rightarrow V \leftarrow \cdots \rightarrow V$ where $V$ and $U$ correspond to observable and unobservable variables respectively. The decomposition of CBNs into c-components has been important in earlier work \cite{tian-pearl} and continues to be an important structural property here. 

\ignore{A point to emphasize is that the algorithm of Theorem \ref{thm:maininf} samples from not just the observable distribution described by $\mathcal{X}$ but also from interventional distributions obtained from $\mathcal{X}$. This is unavoidable as there exist causal models in which certain interventions are {\em non-identifiable},\footnote{Identifiability in causal models is the subject of extensive study \cite{Pearl95, PR95, KM99, Halpern00, Tian02, SP06, HV08}, and there exist procedures to find all non-identifiable interventions in a causal model. Our line of work is orthogonal to this direction.} meaning that they cannot be computed from the observable distribution, and it could be that exactly on these non-identifiable interventions, $\mathcal{X}$ differs greatly from $\mathcal{M}$. Therefore, it is necessary that the testing algorithm be allowed to intervene  on subsets of the observable variables. On the other hand, it is preferable to reduce the number of interventions, as manipulating certain sets of variables may be costly, unethical, impractical or even impossible in practice (see \cite{LondonKadane02} for an instance of a particularly unethical intervention).}

We can use our techniques to extend Theorem \ref{thm:maininf} in several ways:
\begin{enumerate*}
\item[(1)]
In the {\em two-sample testing} problem for causal models, the tester gets access to two unknown causal models $\mathcal{X}$ and $\mathcal{Y}$ on the same set of variables $\mv{V} \cup \mv{U}$ (observable and unobservable resp.). For a given $\eps > 0$, the goal is to distinguish between (i) $\mathcal{X}=\mathcal{Y}$ and (ii) $\Delta(\mathcal{X}, \mathcal{Y}) > \eps$ by sampling from and intervening on $\mv{V}$ in both $\mathcal{X}$ and $\mathcal{Y}$. 

We solve the two-sample testing problem when the inputs are two CBNs over the same DAG $G$ in $n$ variables; for a given $\eps > 0$ and DAG $G$, call the problem $\mathsf{C2ST}(G,\eps)$. Specifically, we show an algorithm to solve $\mathsf{C2ST}(G,\eps)$ that makes $O(\log n)$ interventions on the input models $\mathcal{X}$ and $\mathcal{Y}$, uses $\tilde{O}(n/\eps^2)$ samples per intervention and runs in time $\tilde{O}(n^2/\eps^2)$, when $G$ has bounded in-degree and c-component size.\footnote{Of course, it is allowed for the two networks to be different subgraphs of $G$. So, $\mathcal{X}$ could be defined by the graph $G_1$ and $\mathcal{Y}$ by $G_2$. Our result holds when $G_1 \cup G_2$ is a DAG with bounded in-degree and c-component size.}

\item[(2)]
For the $\mathsf{C2ST}(G,\eps)$ problem, the requirement that $G$ be fully known is rather strict. Instead, suppose  the common graph $G$ is unknown and only bounds on its in-degree and maximum c-component size are given. For example, the biologist above who wants to test whether certain causal mechanisms are identical for patients with and without migraine can reasonably assume that the underlying causal graph is the same (even though he doesn't know what it is exactly) and that only the strengths of the relationships may differ between subjects with and without migraine. For this problem, we obtain an efficient algorithm with nearly the same number of samples and interventions as above.
\ignore{
In this case, the time complexity of the algorithm increases by a factor of $n^{O(d\ell)}$ where $d$ is the maximum in-degree and $\ell$ is the maximum c-component size.}

\item[(3)]
The problem of {\em learning} a causal model can be posed as follows: the learning algorithm gets access to an unknown causal model $\mathcal{X}$ over a set of variables $\mv{V} \cup \mv{U}$ (observable and unobservable resp.), and its objective is to output a causal model $\mathcal{N}$ such that $\Delta(\mathcal{X},\mathcal{N}) \leq \eps$. 
\ignore{We allow the learning algorithm to be {\em improper}, meaning that $\mathcal{N}$ need not satisfy the hypotheses about $\mathcal{X}$.  }

We consider the problem $\clp$ of learning a CBN over a known DAG $G$ on the observable and unobservable variables. For example, this is the problem facing the engineer above who wants to learn a good model for his circuit by conducting some experiments; the DAG $G$ in this case is known from the circuit layout. Given a DAG $G$  with bounded in-degree and c-component size and a parameter $\eps > 0$, we design an algorithm that on getting access to a CBN $\mathcal{X}$ defined over $G$, makes $O(\log n)$ interventions, uses $\tilde{O}(n^2/\eps^4)$ samples per intervention, runs in time $\tilde{O}(n^3/\eps^4)$, and returns an oracle $\mathcal{N}$ that can efficiently compute $P_{\mathcal{X}}[\mv{V}\setminus \mv{T} \mid do(\mv{T}=\mv{t})]$ for any $\mv{T} \subseteq \mv{V}$ and $\mv{t} \in \Sigma^{|\mv{T}|}$ with error at most $\eps$ in TV distance. 

\ignore{The problem of transferring the causal knowledge across multiple source (population) domains to a different target domain with certain commonalities between the domains has been studied extensively under the notion of {\em transportability} \cite{meta-transportability,causal-hierarchy,m-transportability,transportability-introduction,transportability-completeness}.  Such problems require the knowledge about the set of all interventional distributions with respect to each source domain.  For problems such as transportability, our learning algorithm not only minimizes the number of interventions to be performed but also minimizes the number of samples required to learn each interventional distribution to error at most $\eps$.}
\end{enumerate*}
\ignore{An application of our learning algorithm is to the problem of {\em transportability} \cite{meta-transportability,causal-hierarchy,m-transportability,transportability-introduction,transportability-completeness}.  Transportability refers to the notion of transferring causal knowledge from a set of source domains to a target domain, when there are certain commonalities between the source and target domains.  Most work in this area assume the existence of an algorithm that learns the set of {\em all} interventions of the source domains.  Our learning algorithm can be used for this purpose; it is efficient in terms of time, interventions, and sample complexity, and it learns each intervention distribution to error at most $\epsilon$.}

The sample complexity of our testing algorithms matches the state-of-the-art for testing identity of (standard) Bayes nets \cite{CP16, CanonneDKS17}. Designing a goodness-of-fit tester using $o(n)$ samples is a very interesting challenge and seems to require fundamentally new techniques. 

We also show that the number of interventions for $\cit$ and $\clp$ is nearly optimal, even in its dependence on the in-degree and c-component size, and even when the algorithms are allowed to be adaptive.  By `adaptive' we mean the algorithms are allowed to choose the future interventions based on the samples observed from the past interventions.  Specifically,

\begin{theorem}
There exists a causal graph $G$ on $n$ vertices, with maximum in-degree at most $d$ and  largest c-component size at most $\ell$, such that $\Omega(|\Sigma|^{\ell d - 2} \log n)$ interventions are necessary for any algorithm (even adaptive) that solves $\cit$ or $\clp$.
\end{theorem}

\subsection{Related Work} \label{sec:related_work}
\subsubsection{Causality}
As mentioned before, there is a huge and old literature on causality, for both testing causal relationships and inferring causal graphs that is impossible to detail here. Below, we point out some representative directions of research that are relevant to our work. This discussion is far from exhaustive, and the reader is encouraged to pursue the references cited in the mentioned works. 

Most work on statistical tests for causal models has been in the parametric setting. {\em Structural equation models} have traditionally been tested for goodness-of-fit by comparing observed and predicted covariance matrices \cite{BL92}. Another class of tests that has been proposed assumes that the causal factors and the noise factors are conditionally independent. In the {\em additive noise model} \cite{HJMPS09, PJS11, zhang2012kernel, sen2017model}, each variable is the sum of a (non-linear) function of its parent variables and independent noise, often assumed to be Gaussian. This point of view has been refined into an information-geometric criterion in \cite{JM+12}.
 In the non-parametric setting, which is the concern of this paper, Tian and Pearl \cite{tian-pearl} show how to derive functional constraints from causal Bayesian graphs that give equality and inequality constraints among the (distributions of) observed variables, not just conditional independence relations. Kang and Tian \cite{KT06} derive such functional constraints on interventional distributions. Although these results yield non-trivial constraints, they are valid for any model that respects a particular graph and it is not clear how to use them for testing goodness-of-fit with statistical guarantees.

Learning in the context of causal inference has been extensively studied.
To the best of our knowledge, though, most previous work is on learning only the causal graph, whereas our objective is to learn the entire causal model (i.e., the set of all interventional distributions). 
Pearl and Verma \cite{PV95, PV92} investigated the problem of finding a causal graph with hidden variables that is consistent with a given list of conditional independence relations in observational data. In fact, there may be a large number of causal graphs that are consistent with a given set of conditional independence relations. \cite{SGS00, ARSZ05}, and Zhang \cite{Zhang08} (building on the FCI algorithm \cite{SMT99}) has given a complete and sound algorithm for recovering a representative of the equivalence class consistent with a set of conditional independence relations. 

Subsequent work considered the setting when both observational and interventional data are available. 
 This setting has been a recent focus of study \cite{hauser12,wang17,yang18}, motivated by advances in genomics that allow high-resolution observational and interventional data for gene expression using flow cytometry and CRISPR technologies \cite{sachs05,macosko15,dixit16}. When there are no confounding variables, Hauser and B\"uhlmann \cite{HB12}, following up on work by Eberhardt and others \cite{EGS05, Eberhardt07}, find the information-theoretically minimum number of interventions that are sufficient to identify\footnote{More precisely, the goal is to discover the causal graph given the conditional independence relations satisfied by the interventional distributions.} the underlying causal graph and provide a polynomial time algorithm to find such a set of interventions.  A recent paper \cite{KDV17} extends the work of \cite{HB12} to minimize the total cost of interventions where each vertex is assigned a cost.
Another work by Shanmugam et al.~\cite{SKDV15} investigates the problem of learning causal graphs without confounding variables using interventions on sets of small size.  In the presence of confounding variables, there are several works which aim to learn the causal graph from interventional data (e.g., \cite{MMLM06, HEH13}). In particular, a recent work of Kacaoglu et al.~\cite{KSB17} gives an efficient randomized algorithm to learn a causal graph with confounding variables while minimizing the number of interventions from which conditional independence relations are obtained.  

\ignore{All the works mentioned above assume access to an oracle that gives conditional independence relations between variables in the observed and interventional distributions. This is clearly a problematic assumption because it implicitly requires unbounded training data. For example, Scheines and Spirtes \cite{Scheines08} have pointed out that measurement error, quantization and aggregation can easily alter conditional independence relations. The problem of developing finite sample bounds for testing and learning causal models has been repeatedly posed in the literature. The excellent survey by Guyon, Janzing and Sch\"olkopf \cite{GJS10} on causality from a machine learning perspective underlines the issue as one of the ``ten open problems'' in the area. To the best of our knowledge, our work is the first to show finite sample complexity and running time bounds for inference problems on CBNs.
with a quantitative bound on the number of interventions for testing and learning causal Bayesian networks (with or without confounding variables).}

All the works mentioned above assume access to an oracle that gives conditional independence relations between variables in the observed and interventional distributions.  This is clearly a problematic assumption because it implicitly requires unbounded training data. For example, Scheines and Spirtes \cite{Scheines08} have pointed out that measurement error, quantization and aggregation can easily alter conditional independence relations. The problem of developing finite sample bounds for testing and learning causal models has been repeatedly posed in the literature. The excellent survey by Guyon, Janzing and Sch\"olkopf \cite{GJS10} on causality from a machine learning perspective underlines the issue as one of the ``ten open problems'' in the area. To the best of our knowledge, our work is the first to show finite sample complexity and running time bounds for inference problems on causal Bayesian networks.


{\ignore{With respect to the identification of causal effects, that is the effect of interventions, there is a huge amount of literature that consider the setting when the causal graph is known.  }
An application of our learning algorithm is to the problem of {\em transportability}, studied in \cite{meta-transportability,causal-hierarchy,m-transportability,transportability-introduction,transportability-completeness}, which refers to the notion of transferring causal knowledge from a set of source domains to a target domain to identify causal effects in the target domain, when there are certain commonalities between the source and target domains.  Most work in this area assume the existence of an algorithm that learns the set of {\em all} interventions, that is the complete specification of the the source domain model.  Our learning algorithm can be used for this purpose; it is efficient in terms of time, interventions, and sample complexity, and it learns each intervention distribution to error at most $\epsilon$.  \ignore{Our two-sample testing algorithm \footnote{Our algorithm works even for the case when the graphs are unknown but common.} can be used here to test whether the functional relationships among the vertices between different domains (eg: population domains) are equivalent.}

\subsubsection{Distribution Testing and Learning}

There is a vast literature on testing and learning high dimensional distributions in the statistics, and information theory literature, and more recently in computer science with a focus on the computational efficiency of solving such problems. We will not be able to cover and do justice to all of these works in this section. However, we will provide pointers to some of the resources, and also discuss some of the recent progress that is the most closely related to the work we present here. 

In the distribution learning and testing framework, the closest to our work is  learning and testing graphical models. The seminal work of Chow-Liu~\cite{ChowL68} considered the problem of learning tree-structured graphical models. Motivated by applications across many fields, the problem of learning graphical models from samples has gathered recent interest. Of particular interest is the apparent gap between the sample complexity and computational complexity of learning graphical models.~\cite{AbbeelKN06, BreslerMS08} provided algorithms for learning bounded degree graphical models with polynomial sample and time complexity. A lower bound on the sample complexity that grows exponentially with the degree, and only logarithmically with the number of dimensions was provided by~\cite{santhanam2012information}, and recent works~\cite{Bresler15, VuffrayMLC16,KlivansM17} have proposed algorithms with near optimal sample complexity, and polynomial running time for learning Ising models. 

Sample and computational complexity of testing graphical models has been studied recently, in~\cite{CanonneDKS17} for testing Bayesian Networks, and in~\cite{DaskalakisDK18} for testing Ising models. Given sample access to an unknown Bayesian Network, or Ising model, they study the sample complexity, and computation complexity of deciding whether the unknown model is equal to a known fixed model (hypothesis testing). 

The problem of testing and learning distribution properties has itself received wide attention in statistics with a history of over a century~\cite{Fisher25, LehmannR06, CoverT06}. In these fields, the emphasis is on asymptotic analysis characterizing the convergence rates, and error exponents, as the number of samples tends to infinity. A recent line of work originating from~\cite{GoldreichR00, BatuFRSW00} focuses on \emph{sublinear} algorithms where the goal is to design algorithms with the number of samples that is smaller than the domain size (e.g.,~\cite{Canonne15, Goldreich17}, and references therein). 

While most of these results are for learning and testing low dimensional (usually one dimensional) distributions, there are some notable exceptions. Testing for properties such as independence, and monotonicity in high dimensions have been considered recently~\cite{BhattacharyyaFRV11, AcharyaDK15, DiakonikolasK16}. These results show that the optimal sample complexity for testing these properties grows exponentially with the number of dimensions. A line of recent work~\cite{CP16, CanonneDKS17, DaskalakisDK17, DaskalakisDK18} overcomes this barrier by utilizing additional structure in the high-dimensional distribution induced by Bayesian network or Markov Random Field assumptions.

\subsection{Overview of our Techniques} 
In this section, we give an overview of the proof of Theorem \ref{thm:maininf} and the lower bound construction. We start by making a well-known observation \cite{tian-pearl,verma-pearl} that CBNs can be assumed to be over a particular class of DAGs known as {\em semi-Markovian causal graphs}. A semi-Markovian causal graph is a DAG where every vertex corresponding to an unobservable variable is a root and has exactly two children, both observable. More details of the correspondence are given in Appendix \ref{section:general-graph-reduction}. 

In a semi-Markovian causal graph, two observable vertices $V_1$ and $V_2$ are said to be connected by a bi-directed edge if there is a common unobservable parent of $V_1$ and $V_2$. Each connected component of the graph restricted to bi-directed edges is called a {\em c-component}. The decomposition into c-components gives very useful structural information about the causal model. In particular, a fact that is key to our whole analysis is that if $\mathcal{N}$ is a semi-Markovian Bayesian network on observable and unobservable variables $\mv{V} \cup \mv{U}$ with c-components $\mv{C}_1, \dots, \mv{C}_p$, then for any $\mv{v} \in  \Sigma^{|\mv{V}|}$:
\begin{align}\label{eqn:factor}
P_{\mathcal{N}}[\mv{v}] = \prod_{i=1}^p P_{\mathcal{N}}[\mv{c}_i \mid do(\mv{V}\setminus \mv{C}_i = \mv{v}\setminus \mv{c}_i)]
\end{align}
where {$\Sigma$ is the alphabet set,} $\mv{c}_i$ is the restriction of $\mv{v}$ to $\mv{C}_i$ and $\mv{v}\setminus \mv{c}_i$ is the restriction of $\mv{v}$ to $\mv{V}\setminus \mv{C}_i$. Moreover, one can write a similar formula  (Lemma \ref{lemma:c-component-factorization}) for an interventional distribution on $\mathcal{N}$ instead of the observable distribution $P_{\mathcal{N}}[\mv{v}]$.

The most direct approach to test whether two causal Bayes networks $\mathcal{X}$ and $\mathcal{Y}$ are identical is to test whether each interventional distribution is identical in the two models.  This strategy would require $(|\Sigma|+1)^n$ many interventions, each on a variable set of size $O(n)$, where $n$ is the total number of observable vertices.
To reduce the number of interventions as well as the sample complexity, 
a natural approach, given \eqref{eqn:factor} and its extension to interventional distributions, is to test for identity between each pair of ``local'' distributions $$P_{\mathcal{X}}[\mv{S} \mid do(\mv{v}\setminus \mv{s})] \qquad \text{and} \qquad P_{\mathcal{Y}}[\mv{S} \mid do(\mv{v}\setminus \mv{s})]$$ for every subset $\mv{S}$ of a c-component $\mv{C}$ and assignment $\mv{v}\setminus \mv{s}$ to $\mv{V}\setminus \mv{S}$. We assume that each c-component is bounded, so each local distribution has bounded support. Moreover, using the conditional independence properties of Bayesian networks, note that in each local distribution, we only need to intervene on observable parents of $\mv{S}$ that are outside $\mv{S}$, not on all of $\mv{V} \setminus \mv{S}$.  

Through a probabilistic argument, we efficiently find a \emph{small set} $\mv{I}$ of {\em covering interventions}, which are defined as a set of interventions with the following property: For every subset $\mv{S}$ of a c-component and for every assignment $\pa(\mv{S})$ to the observable parents of $\mv{S}$, there is an intervention $I \in \mv{I}$ that does not intervene on $\mv{S}$ and sets the parents of $\mv{S}$ to exactly $\pa(\mv{S})$. Our test performs all the interventions in $\mv{I}$ on both $\mathcal{X}$ and $\mathcal{Y}$ and hence can observe each of the local distributions $P_{\mathcal{X}}[\mv{S} \mid do(\pa(\mv{S}))]$ and $P_{\mathcal{Y}}[\mv{S} \mid do(\pa(\mv{S}))]$.
What remains is to bound $\Delta(\mathcal{X},\mathcal{Y})$ in terms of the distances between each pair of local distributions. 

To that end, we develop a subadditivity theorem about CBNs, and this is the main technical contribution of our upper bound results.  We show that if each pair of local distributions is within distance $\gamma$ in {\em squared Hellinger} distance, then for any intervention $I$, applying $I$ to $\mathcal{X}$ and $\mathcal{Y}$ results in distributions that are within $O(n \gamma)$ distance in {squared Hellinger} distance, assuming bounded in-degree and c-component size of the underlying graph. A bound on the total variation distance between the interventional distributions and hence $\Delta(\mathcal{X},\mathcal{Y})$ follows.  The subadditivity theorem is inspired from \cite{CP16}, where they showed that for Bayes networks, ``closeness of local marginals implies closeness of the joint distribution''. Our result is in a very different set-up, where we prove ``closeness of local interventions implies closeness of any joint interventional distribution'', and requires a new proof technique.
We relax the squared Hellinger distance between the interventional distributions as the objective of a minimization program in which the constraints are that each pair of local distributions is $\gamma$-close in squared Hellinger distance. By a sequence of transformations of the program, we lower bound its objective in terms of $\gamma$, thus proving our result. {In the absence of unobservable variables, the analysis becomes much simpler and is sketched in Appendix~\ref{sec:fully_observable_case}.}

{Regarding the lower bound, we prove that the number of interventions required by our algorithms are indeed necessary for any algorithm that solves $\cit$ or $\clp$, even if the algorithms are provided with infinite samples/time.  For any algorithm that fails to perform some local intervention $I$, we provide a construction of two models which do not agree on $I$ and agree on all other interventions.  Our construction is designed in such a way that it allows adaptive algorithms.  The idea is to show an adversary that, for each intervention, reveals a distribution to the algorithm.  Towards the end, when the algorithm fails to perform some local intervention $I$, we can show a construction of two models such that: i) both the models do not agree on $I$, and the total variation distance between the interventional distributions is equal to one; ii) and for all other interventions, the interventional distributions revealed by the adversary match with the corresponding distributions on both the models.  This, together with a probabilitic argument, shows the existence of a causal graph that requires sufficiently large number of interventions to solve $\cit$ and $\clp$.}

\subsection{Future Directions}

We hope that this work paves the way for future research on designing efficient algorithms with bounded sample complexity for learning and testing causal models. For the sake of concreteness, we list a few open problems.

\begin{itemize}
\item
Interventional experiments are often expensive or infeasible, so one would like to deduce causal models from observational data alone. In general, this is impossible. However, in {\em identifiable} causal Bayesian networks (see \cite{Tian02}), one can identify causal effects from observational data alone. \textbf{Is there an efficient algorithm to learn an identifiable interventional distribution from samples?}\footnote{Schulman and Srivastava \cite{SS16} have shown that under adversarial noise, there exist causal Bayesian networks on $n$ nodes where estimating an identifiable intervention to precision $d$ requires precision $d+\exp(n^{0.49})$ in the estimates of the probabilities of observed events. However, this instability is likely due to the adversarial noise and does not preclude an efficient sampling-based algorithm, especially if we assume a balancedness condition as in \cite{CanonneDKS17}.}

\item
A deficiency of our work is that we assume the underlying causal graph is fully known. 
\textbf{Can our learning algorithm be extended to the setting where the hypothesis only consists of some limited information about the causal graph (e.g., in-degree, c-component size) instead of the whole graph?} In fact, it is open how to efficiently learn the distribution given by a Bayesian network based on samples from it, if we don't have access to the underlying graph \cite{CP16, CanonneDKS17}.

\item
Our goodness-of-fit algorithm might reject even when the input $\mathcal{X}$ is very close to the hypothesis $\mathcal{M}$. \textbf{Is there a {\em tolerant} goodness-of-fit tester that accepts when $\Delta(\mathcal{X}, \mathcal{M}) \leq \eps_1$ and rejects when $\Delta(\mathcal{X}, \mathcal{M}) > \eps_2$ for $0 < \eps_1 < \eps_2 < 1$?} Our current analysis does not extend to a tolerant tester. The same question holds for two-sample testing.

\item
In many applications, causal models are described in terms of {\em structural equation models}, in which each variable is a deterministic function of its parents as well as some stochastic error terms. \textbf{Design sample and time efficient algorithms for testing and learning structural equation models.} Other questions such as evaluating {\em counterfactual} queries  or doing {\em policy analysis} (see Chapter 7 of \cite{Pearl00}) also present interesting algorithmic problems.
\end{itemize}

\section{Preliminaries}\label{sec:prelims}

\paragraph{Notation.} 
We use capital (bold capital) letters to denote variables (sets of variables), e.g., $A$ is a variable and $\mv{B}$ is a set of variables. We use small (bold small) letters to denote values taken by the corresponding variables (sets of variables), e.g., $a$ is the value of $A$ and $\mv{b}$ is the value of the set of variables $\mv{B}$. 
The variables in this paper take values in a discrete set $\Sigma$. We use $[n]$ to denote $\{1,2, \dots, n\}$. 

\paragraph{Probability and Statistics.}
The total variation (TV) distance between distributions $P$ and $Q$ over the same set $[D]$ is $\delta_{TV}(P,Q):=\frac12 \sum_{i\in [D]} |P(i)-Q(i)|.$
The squared Hellinger distance (given in~\eqref{eqn:squaredHellinger}) and the total variation distance are related by the following.

\begin{lemma}[Hellinger vs total variation] \label{Hellinger-Tv-Inequality} The Hellinger distance and the total variation distance between two distributions $P$ and $Q$ are related by the following inequality: $$ H^2(P,Q) \leq \delta_{TV}(P,Q) \leq \sqrt{2 H^2(P,Q)}.$$
\end{lemma}


The problem of two-sample testing for discrete distributions in Hellinger distance, and learning with respect to total variation distance has been studied in the literature, and the following two lemmas state two results we use. Let $P$ and $Q$ denote distributions over a domain of size $D$. 


\begin{lemma}[Hellinger Test,~\cite{DiakonikolasK16}]
\label{hellingerTest} 
Given $\tilde{O}(\min (D^{2/3}/\epsilon^{8/3},D^{3/4}/\epsilon^2))$ samples from each unknown distributions $P$ and $Q$, we can distinguish between $P=Q$ vs $H^2(P,Q)\geq \epsilon^2$ with probability at least $2/3$.  This probability can be boosted to $1-\delta$ at a cost of an additional  $O(\log (1 / \delta))$ factor in the sample complexity. The running time of the algorithm is quasi-linear in the sample size.
\end{lemma}


\begin{lemma}[Learning in TV distance, folklore (e.g.~\cite{DevroyeLugosi})]\label{lemma:TVlearn}
For all $\delta \in (0,1)$, the empirical distribution $\hat{P}$ computed using $\Theta \left(\frac{D}{\epsilon^2} + {\log{\frac{1}{\delta}} \over \epsilon^2}\right)$ samples from $P$  satisfies $H^2(P,\hat{P}) \le \delta_{TV}(P,\hat{P}) \le \epsilon$, with probability at least $1-\delta$. 
\end{lemma}

\paragraph{Bayesian Networks.}
Bayesian networks are popular probabilistic graphical models for describing high-dimensional distributions.

\begin{definition}
A {\em Bayesian Network} ($\bn$) $\mathcal{N}$ is a distribution that can be specified by a tuple $\langle \mv{V}, G, \{\Pr[V_i \mid \pa(V_i)]: V_i \in \mv{V}, \pa(V_i) \in \Sigma^{|\Pa(V_i)|}\} \rangle$ where: (i) $\mv{V}$ is a set of variables over alphabet $\Sigma$, (ii) $G$ is a directed acyclic graph with nodes corresponding to the elements of $\mv{V}$, and (iii) $\Pr[V_i \mid \pa(V_i)]$ is the conditional distribution of variable $V_i$ given that its parents $\Pa(V_i)$ in $G$ take the values $\pa(V_i)$. 

The Bayesian Network 
$\mathcal{N} =  \langle \mv{V}, G, \{\Pr[V_i \mid \pa(V_i)]\} \rangle$ 
defines a unique probability distribution $P_{\mathcal{N}}$ over $\Sigma^{|\mv{V}|}$, as follows. For all $\mv{v} \in \Sigma^{|\mv{V}|}$,
$$P_{\mathcal{N}}[\mv{v}] = \prod_{V_i \in \mv{V}} \Pr[v_i  \mid  \pa(V_i)].$$
In this distribution, each variable $V_i$ is independent of its non-descendants given its parents in $G$. 
\end{definition}

Conditional independence relations in graphical models are captured by the following definitions.

\begin{definition}
Given a DAG $G$, a (not necessarily directed) path $p$ in $G$ is said to be {\em blocked} by a set of nodes $\mathbf{Z}$, if (i) $p$ contains a chain node $B$ ($A \rightarrow B \rightarrow C$) or a fork node $B$ ($A \leftarrow B \rightarrow C$) such that $B \in \mathbf{Z}$ (or) (ii) $p$ contains a collider node $B$ ($A \rightarrow B \leftarrow C$) such that $B \notin \mathbf{Z}$ and no descendant of $B$ is in $\mathbf{Z}$. 
\end{definition}

\begin{definition}[d-separation]
For a given DAG $G$ on $\mv{V}$, two disjoint sets of vertices $\mathbf{X,Y} \subseteq \mv{V}$ are said to be {\em d-separated} by $\mathbf{Z}$ in $G$, if every (not necessarily directed) path in $G$ between $\mathbf{X}$ and $\mathbf{Y}$ is blocked by $\mathbf{Z}$.
\end{definition}

\begin{lemma}[Graphical criterion for independence] 
For a given $\bn$ $\mathcal{N} = \langle \mv{V}, G, \{\Pr[V_i  \mid  \pa(V_i)]\}\rangle$ and $\mv{X,Y,Z} \subset \mv{V}$, if $\mathbf{X}$ and $\mathbf{Y}$ are {d-separated} by $\mathbf{Z}$ in $G$, then $\mathbf{X}$ is {\em independent} of $\mathbf{Y}$ given $\mathbf{Z}$ in $P_{\mathcal{N}}$, denoted by $[\mathbf{X} \indep \mathbf{Y}  \mid  \mathbf{Z}]$ in $P_{\mathcal{N}}$.
\end{lemma}             

\subsection{Causality}
We describe Pearl's notion of causality from~\cite{Pearl95}. Central to his formalism is the notion of an {\em intervention}. Given a variable set $\mv{V}$ and a subset $\mv{X} \subset \mv{V}$, an intervention $do(\mv{x})$ is the process of fixing the set of variables $\mv{X}$ to the values $\mv{x}$. The {\em interventional distribution} $\Pr[\mv{V} \mid  do(\mv{x})]$ is the distribution on $\mv{V}$ after setting $\mv{X}$ to $\mv{x}$. As discussed in the introduction, an intervention is quite different from conditioning.

Another important component of Pearl's formalism is that some variables may be unobservable. The unobservable variables can neither be observed nor be intervened.  We partition our variable set into two sets $\mv{V}$ and $\mv{U}$, where the variables in $\mv{V}$ are {\em observable} and the variables in $\mv{U}$ are {\em unobservable}. Given a directed acyclic graph $H$ on $\mv{V \cup U}$ and a subset $\mv{X} \subseteq (\mv{V \cup U})$, we use $\bm{\Pi}_H(\mv{X}), \Pa_H(\mv{X})$, $\An_H(\mv{X})$, and $\De_H(\mv{X})$  to denote the set of all parents, observable parents, observable ancestors and observable descendants respectively of $\mv{X}$, excluding $\mv{X}$, in $H$. When the graph $H$ is clear, we may omit the subscript. As usual, small letters, $\bm{\pi}(\mv{X}),$ $\pa(\mv{X})$, $\an(\mv{X})$ and $\de(\mv{X})$  are used to denote their corresponding values. And, we use $H_{\overline{\mv{X}}}$ and $H_{\underline{\mv{X}}}$ to denote the graph obtained from $H$ by removing the incoming edges to $\mv{X}$ and outgoing edges from $\mv{X}$ respectively.  

\begin{definition}[Causal Bayesian Network] \label{def:causal Bayesnet}
A {\em causal Bayesian network} (\cbn) is a collection of interventional distributions that can be defined in terms of a tuple $\langle \mv{V}, \mv{U}, G,$ $\{\Pr[V_i \mid \bm{\pi}(V_i)] : V_i \in \mv{V}, \bm{\pi}(V_i) \in \Sigma^{|\bm{\Pi}(V_i)|}\}, \{\Pr[U_i \mid \bm{\pi}(U_i)] : U_i \in \mv{U}, \bm{\pi}(U_i) \in \Sigma^{|\bm{\Pi}(U_i)|}\}\rangle$, where (i) $\mv{V}$ and $\mv{U}$ are the sets of observable and unobservable variables respectively, (ii) $G$ is a directed acyclic graph on $\mv{V} \cup \mv{U}$, and (iii) $\Pr[V_i \mid \bm{\pi}(V_i)]$  and $\Pr[U_i \mid \bm{\pi}(U_i)]$ are the conditional probability distributions of $V_i$  and $U_i$ resp.~given that its parents $\bm{\Pi}(V_i)$ and $\bm{\Pi}(U_i)$ resp.~take the values $\bm{\pi}(V_i)$ and $\bm{\pi}(U_i)$) resp. 

A \cbn\ 
$\mathcal{M} = \langle \mv{V}, \mv{U}, G,$ $\{\Pr[V_i \mid \bm{\pi}(V_i)]\}, \{\Pr[U_i \mid \bm{\pi}(U_i)]\}\rangle$ 
defines a unique interventional distribution $P_{\mathcal{M}}[\mv{V}  \mid  do(\mv{x})]$ for every subset $\mv{X} \subseteq \mv{V}$ (including $\mv{X} = \emptyset$) and assignment $\mv{x} \in \Sigma^{|\mv{X}|}$, as follows. For all $\mv{v} \in \Sigma^{|\mv{V}|}$:
$$P_{\mathcal{M}}[\mv{v}  \mid  do(\mv{x})] = 
\begin{cases}
\sum_{\mv{u}} \prod_{V_i \in \mv{V}\setminus \mv{X}} \Pr[v_i \mid \bm{\pi}(V_i)] \cdot \prod_{U_i \in \mv{U}} \Pr[u_i  \mid \bm{\pi}(U_i)]
& \text{if }\mv{v} \text{is consistent with }\mv{x}\\
0 & \text{otherwise.}
\end{cases}
$$
We say that $G$ is the {\em causal graph} corresponding to the {\cbn} $\mathcal{M}$.
\end{definition}

Another equivalent way to define a {\cbn} is by specifying the set of interventional distributions $P_{\mathcal{M}}[\mv{V}  \mid  do(\mv{x})]$ for all subsets $\mv{X}$ and assignments $\mv{x}$. To connect to the preceding definition, we require that each $P_{\mathcal{M}}[\mv{V}\mid do(\mv{x})]$ is defined by the Bayesian network described by $G_{\overline{\mv{X}}}$ with the conditional probability distributions obtained by setting the variables in $\mv{X}$ to the constants $\mv{x}$. 

It is standard in the causality literature to work with causal graphs of a particular structure:
\begin{definition}[Semi-Markovian causal graph and Semi-Markovian Bayesian network] \label{defn:SMCG}
A {\em semi-Markovian causal graph} (\smcg) $G$ is a directed acyclic graph on $\mv{V \cup U}$ where every unobservable variable is a root node and has exactly two children, both observable. A {\em semi-Markovian Bayesian network (\smbn)} is a causal Bayesian network where the causal graph is semi-Markovian. 
\end{definition}
There exists a known reduction (described formally in Appendix~\ref{section:general-graph-reduction}) from general causal Bayesian networks to semi-Markovian Bayesian networks that preserves all the properties we use in our analysis, so that henceforth, we will restrict only to \smbn s.

In \smcg s, the divergent edges $V_i \leftarrow U_k \rightarrow V_j$ are usually represented by {\em bi-directed edges} $V_i \leftrightarrow V_j$.  A bi-directed edge between two observable variables implicitly represents the presence of an unobservable parent. 
\begin{definition}[c-component] \label{defn:c-component}
For a given {\smcg} $G$, $\mv{S} \subseteq \mv{V}$ is a {\em c-component} of $G$, if $\mv{S}$ is a maximal set such that between any two vertices of $\mv{S}$ , there exists a path that uses only bi-directed edges.
\end{definition}

Since a c-component forms an equivalence relation, the set of all c-components forms a partition of $\mv{V}$, the observable vertices of $G$.  We use the notation $C(\mv{V})=\{\mv{S}_1,\mv{S}_2,\ldots,\mv{S}_k\}$ to denote the partition of $\mv{V}$ into the c-components of $G$, where each $\mv{S}_i \subseteq \mv{V}$ is a c-component of $G$.  

Also, for $\mv{X} \subseteq \mv{V}$, the induced subgraph $G[\mv{X}]$ is the subgraph obtained by removing the vertices $\mv{V \setminus X}$ and their corresponding edges from $G$.  We use the notation $C(\mv{X}) = \{\mv{S}_1,\mv{S}_2,\ldots,\mv{S}_k\}$ to denote the set of all c-components of $G[\mv{X}]$, that is each $\mv{S}_i \subseteq \mv{X}$ is a c-component of $G[\mv{X}]$. The next two lemmas capture the factorizations of distributions in \smbn. 

\begin{lemma}\label{lemma:trimUnnecessaryActions}
Let $\mathcal{M}$ be a given {\smbn} with respect to the {\smcg} $G$.  For any set $\mv{S}\subseteq \mv{V}$, and a subset $\mv{D}$ such that $(\mv{V} \setminus \mv{S}) \supseteq \mv{D} \supseteq \Pa(\mv{S})$, and for any assignment $\mv{s,d}$,
$$
P_{\mathcal{M}}[\mv{s}  \mid  do(\mv{d})] = P_{\mathcal{M}}[ \mv{s}  \mid  do(\pa(\mv{S}))]
$$
where $\pa(\mv{S})$ is consistent with the assignment $\mv{d}$.
\end{lemma}        

\begin{proof}
When the parents of $\mv{S}$, $\Pa(\mv{S})$, are targeted for intervention, the distribution on $\mv{S}$ remains the same irrespective of whether the other vertices in $(\mv{V} \setminus \mv{S})$ are intervened or not. 
\end{proof}

\begin{lemma}[c-component factorization, \cite{tian-pearl}] \label{lemma:c-component-factorization}
Given a {\smbn} $\mathcal{M}$ with respect to the causal graph $G$ and a subset $\mathbf{X} \subseteq \mathbf{V}$, let $C(\mathbf{V \backslash X}) = \{\mathbf{S}_1, \ldots ,\mathbf{S}_k\}$.  For any given assignment $\mathbf{v}$, $$P_{\mathcal{M}}[\mathbf{v \setminus x} \mid do(\mathbf{x})] = \prod_i P_{\mathcal{M}}[\mathbf{s}_i  \mid  do(\mathbf{v \setminus s}_i)].$$
\end{lemma}

For a given {\smcg} $G$, the in-degree and out-degree of an observable vertex $V_i \in \mv{V}$ denote the number of observable parents and observable children of $V_i$ in $G$ respectively.  The maximum in-degree of a {\smcg} $G$ is the maximum in-degree over all the observable vertices.  The maximum degree of a {\smcg} $G$ is the maximum of the sum of the in-degree and out-degree over all the observable vertices.

\begin{definition}[Graphs with {\em bounded in-degree} and {\em bounded c-component}]
$\mathcal{G}_{d,\ell}$ denotes the class of {\smcg}s with maximum in-degree at most $d$ and the size of the largest c-component at most $\ell$.
\end{definition}

\subsection{Problem Definitions}
\noindent Here we define the testing and learning problems considered in the paper.  Let $\mathcal{M}$ and $\mathcal{N}$ be two \smbn s.  We say that $\mathcal{M} = \mathcal{N}$, if   
\begin{align*}
P_{\mathcal{M}}[\mv{V} \setminus \mv{T} \mid do(\mv{t})] = P_{\mathcal{N}}[\mv{V} \setminus \mv{T} \mid do(\mv{t})] \qquad \forall \mv{T} \subseteq \mv{V}, \mv{t} \in \Sigma^{|\mv{T}|}.
\end{align*}

\noindent And we say that $\Delta(\mathcal{M},\mathcal{N}) > \epsilon $, if there exists $\mv{T} \subseteq \mv{V}$ and $\mv{t} \in \Sigma^{|\mv{T}|}$ such that 
\begin{align*}
\delta_{TV}(P_{\mathcal{M}}[\mv{V} \setminus \mv{T} \mid do(\mv{t})], P_{\mathcal{N}}[\mv{V} \setminus \mv{T} \mid do(\mv{t})]) > \epsilon.
\end{align*}

\begin{definition}[Causal Goodness-of-fit Testing ($\cgft$)]
Given a {\smcg} $G$, a (known) {\smbn} $\mathcal{M}$ on $G$, and $\epsilon > 0$. Let $\mathcal{X}$ denote an unknown {\smbn} on $G$.  The objective of $\cgft$ is to distinguish between $\mathcal{X} = \mathcal{M}$ versus $\Delta(\mathcal{X},\mathcal{M}) > \epsilon$ with probability at least 2/3, by performing interventions and taking samples from the resulting interventional distributions of $\mathcal{X}$.
\end{definition}

\begin{definition}[Causal Two-sample Testing ($\cit$)]
Given a {\smcg} $G$, and $\epsilon > 0$.  Let  $\mathcal{X}$ and $\mathcal{Y}$ be two unknown \smbn s on $G$. The objective of $\cit$ is to distinguish between $\mathcal{X}=\mathcal{Y}$ versus $\Delta(\mathcal{X},\mathcal{Y}) > \epsilon$ with probability at least 2/3, by performing interventions and taking samples from the resulting interventional distributions of $\mathcal{X}$ and $\mathcal{Y}$.
\end{definition}

\begin{definition}[Learning SMBNs ($\clp$)]\label{def:learn}
Given a {\smcg} $G$ and $\epsilon > 0$.  Let $\mathcal{X}$ be an unknown  {\smbn} on $G$.  The objective of $\clp$ is to perform  interventions and taking samples from the resulting interventional distributions of $\mathcal{X}$, and return an oracle that for any $\mv{T} \subseteq \mv{V}$ and $\mv{t} \in \Sigma^{|\mv{T}|}$ returns an estimated interventional distribution $P_{ES}[\mv{V} \setminus \mv{T} \mid do(\mv{t})]$ such that \[ \delta_{TV}([P_{\mathcal{X}}[\mv{V} \setminus \mv{T} \mid do(\mv{t})], P_{ES}[\mv{V} \setminus \mv{T} \mid do(\mv{t})]) < \epsilon.\]
\end{definition}

We emphasize that in all three problems, the causal graph $G$ is known explicitly in advance. 

\section{Testing and Learning Algorithms for \smbn s} \label{sec:main results}
Before we discuss our algorithms, we begin by defining \emph{covering intervention sets.} 

\begin{definition} \label{defn:covering_interventions}
A set of interventions $\mv I$ is a \emph{covering intervention set} if for every subset $\mv S$ of every c-component, and every assignment $\pa(\mv S)\in \Sigma^{|\Pa(\mv S)|}$ there exists an $I\in \mv I$ such that,
\begin{itemize}
\item
No node in $\mv S$ is intervened in $I$. 
\item
Every node in $\Pa(S)$ is intervened.
\item
$I$ restricted to $\Pa(S)$ has the assignment $\pa(S)$. 
\end{itemize}
\end{definition}

\noindent Our algorithms comprise of two key arguments.
\begin{itemize}
\item 
A procedure to compute a covering intervention set ${\mv I}$ of \emph{small size}.
\item
A sub-additivity result for \cbn s that allows us to localize the distances: where we show that two \cbn s are far implies there exist a marginal distribution of some intervention in ${\mv I}$ such that the marginals are far.
\end{itemize}
These two results are formalized in Section~\ref{sec:covering-sets}, and Section~\ref{sec:aub-additivity} respectively.
\subsection{Testing}
Our main testing result is the following upper bound for testing of causal models.

\begin{theorem}[Algorithm for $\cit$] \label{outTestingResult} 
Let $G$ be a {\smcg} $\in \mathcal{G}_{d,\ell}$ with $n$ vertices. Let the variables take values over a set $\Sigma$ of size $K$.  Then, there is an algorithm to solve $\cit$,
that makes $O(K^{\ell d} (3d)^\ell \log n )$ interventions to each of the unknown {\smbn}s $\mathcal{X}$ and $\mathcal{Y}$, taking $\tilde{O}(K^{\ell(d+7/4)} n {\epsilon^{-2}})$ samples per intervention, in time $\tilde{O}(2^{\ell} K^{\ell(2d+7/4)} n^2 \epsilon^{-2})$.

When the maximum degree (in-degree plus out-degree) of $G$ is bounded by $d$, then our algorithm uses $O(K^{\ell d} (3d)^\ell \ell d^2 \log K)$ interventions with the same sample complexity and running time as above.
\end{theorem}

This result gives Theorem~\ref{thm:maininf} as a corollary, since two sample tests are harder than one sample tests. 

\begin{proof}[Proof of Theorem~\ref{outTestingResult}]
Our algorithm is described in Algorithm~\ref{alg:1}. 

The algorithm starts with a covering intervention set $\mv I$. Lemma~\ref{lemma:counting_argument} gives an $\mv I$ with $O(K^{\ell d}(3d)^\ell(\log n+\ell d \log K))$ interventions.  When the maximum degree is bounded by $d$, then Lemma~\ref{lemma:counting_lovaszLocal} gives an $\mv I$ of size $O(K^{\ell d} (3d)^\ell \ell d^2 \log K)$. Moreover, by the remarks following Lemmas \ref{lemma:counting_argument} and \ref{lemma:counting_lovaszLocal}, $\mv I$ can be found in $\tilde{O}(n)$ time.

\begin{center}
\framebox{
\begin{minipage}[ht]{0.9\linewidth}
\captionof{algocf}{Algorithm for $\cit$}
\label{alg:1}
$\mv I$: Covering intervention set
\begin{enumerate}
\item 
Under each intervention $I \in {\mv I}$:
\begin{enumerate}
\item
Obtain $\tilde{O}(K^{\ell(d+7/4)} n  \epsilon^{-2})$ samples from the interventional distribution of $I$ in both models $\mathcal{X}$ and $\mathcal{Y}$. 
\item
For any subset $\mv{S}$ of a c-component of $G$, if $I$ does not set $\mv{S}$ but sets $\Pa(\mv{S})$ to $\pa(\mv{S})$, then using Lemma~\ref{hellingerTest}, Lemma~\ref{lemma:trimUnnecessaryActions} and the obtained samples, test (with error probability at most $1/(3 K^{\ell d} 2^{\ell} n )$):
\begin{align*}
P_{\mathcal{X}}[\mv{S} \mid do(\pa(\mv{S}))] = P_{\mathcal{Y}}[\mv{S}  \mid  do(\pa(\mv{S}))]  \text { versus }
H^2\left(\begin{array}{l} P_{\mathcal{X}}[\mv{S} \mid do(\pa(\mv{S}))], \\ \newline P_{\mathcal{Y}}[\mv{S}  \mid  do(\pa(\mv{S}))] \end{array}\right) \geq \dfrac{\epsilon^2}{2 K^{\ell (d+1)} n}
\end{align*}
Output ``$\Delta(\mathcal{X},\mathcal{Y})>\eps$'' if the latter.
\end{enumerate}
\item
Output ``$\mathcal{X}=\mathcal{Y}$''.
\end{enumerate}
\end{minipage}
}
\end{center}

We will now analyze the performance of our algorithm.
 
\medskip
\noindent\textbf{Number of interventions, time, and sample requirements.} The number of interventions is the size of $\mv I$, bounded from Lemma~\ref{lemma:counting_argument} or Lemma~\ref{lemma:counting_lovaszLocal}. The number of samples per intervention is given in the algorithm. The algorithm performs $n 2^{\ell} K^{\ell d}$ sub-tests.  And for each such sub-test, the algorithm's running time is quasi-linear in the sample complexity (Lemma \ref{hellingerTest}), therefore taking a total time of $\tilde{O}(2^{\ell} K^{\ell(2d+7/4)} n^2  \epsilon^{-2})$. 

\medskip
\noindent\textbf{Correctness.} In Theorem~\ref{theorem:mainTheorem}, we show that when $\Delta(\mathcal{X}, \mathcal{Y})>\eps$, there exists a subset $\mv S$ of some c-component, and an $I\in \mv I$ that does not intervene any node in $\mv S$ but intervenes $\Pa(\mv{S})$ with some assignment $\pa(\mv{s})$ such that
\[
H^2(P_{\mathcal{X}}[\mv{S} \mid do(\pa(\mv{S}))], P_{\mathcal{Y}}[\mv{S} \mid do(\pa(\mv{S}))]) > \epsilon^2/(2 K^{\ell (d+1)} n).
\]
This structural result is the key to our algorithm.  This together with Lemma~\ref{Hellinger-Tv-Inequality} proves that $P_{\mathcal{X}}$ and $P_{\mathcal{Y}}$ are far in terms of the total variation distance.  To bound the error probability, note that the number of total sub-tests we run is bounded by $K^{\ell d} n 2^{\ell}$, and the error probability for each subset is at most $1/(3 K^{\ell d} 2^{\ell} n)$, by the union bound, we will have an error of at most $1/3$ over the entire algorithm.
\end{proof} 

In some cases, the underlying $\smcg$ might not be known. We will now consider the problem of two sample testing, where $\mathcal{X}$ and $\mathcal{Y}$ are still on the same common $\smcg$ $G$, but $G$ is unknown. We now show an algorithm that uses the same number of interventions and samples as Theorem~\ref{outTestingResult} for the known $G$ case, however requiring $O(n^{\ell+1}K^{\ell (2d+7/4)} \epsilon^{-2})$ time. 


\begin{theorem}[Algorithm for $\cit$ -- Unknown graph] \label{outTestingResult:unknownGraph} 
Consider the same set-up as Theorem~\ref{outTestingResult}, except that the $\smcg$ $G\in\mathcal{G}_{d,\ell}$ is unknown.
Then, there is an algorithm to this problem, that makes $O(K^{\ell d} (3d)^\ell \log n )$ interventions to $\mathcal{X}$ and $\mathcal{Y}$, taking $\tilde{O}(K^{\ell(d+7/4)} n {\epsilon^{-2}})$ samples per intervention, in time $\tilde{O}({n}^{\ell} K^{\ell(2d+7/4)} n \epsilon^{-2})$.
\end{theorem}

\begin{proof}
We first use Lemma~\ref{lemma:counting_argument} and obtain a set of interventions $\mv{I}$, such that $\mv{I}$ is a covering set with error probability at most $1/6$.  Note that Lemma~\ref{lemma:counting_argument} holds even when the underlying graph $G$ is unknown.    
\begin{center} 
\framebox{
\begin{minipage}[ht]{0.9\linewidth}
\captionof{algocf}{Algorithm for $\cit$ -- Unknown graph}
$\mv I$: Covering intervention set
\begin{enumerate}
\item 
Under each intervention $I = \Pr[\mv{V \setminus T} \mid do(\mv{t})] \in {\mv I}$:
\begin{enumerate}
\item
Obtain $\tilde{O}(K^{\ell(d+7/4)} n  \epsilon^{-2})$ samples from the interventional distribution of $I$ in both models $\mathcal{X}$ and $\mathcal{Y}$.
\item
For each subset $\mv{S} \subseteq \mv{V \setminus T}$ of size $\leq \ell$, using Lemma~\ref{hellingerTest}, Lemma~\ref{lemma:trimUnnecessaryActions} and the obtained samples, test (with error probability at most $1/(6 K^{\ell d} 2^{\ell} n )$):
\begin{align*}
P_{\mathcal{X}}[\mv{S} \mid do(\mv{t})] = P_{\mathcal{Y}}[\mv{S}  \mid  do(\mv{t})] \quad \text {versus } \quad
H^2\left(\begin{array}{l} P_{\mathcal{X}}[\mv{S} \mid do(\mv{t})], \\ \newline P_{\mathcal{Y}}[\mv{S}  \mid  do(\mv{t})] \end{array}\right) \geq \dfrac{\epsilon^2}{2 K^{\ell (d+1)} n}
\end{align*}
Output ``$\Delta(\mathcal{X},\mathcal{Y})>\eps$'' if the latter.
\end{enumerate}
\item
Output ``$\mathcal{X}=\mathcal{Y}$''.
\end{enumerate}
\end{minipage}
}
\end{center}
For each intervention, we go over all subsets $S$ of size $\leq \ell$.  Therefore we perform at most $\binom{n}{\leq \ell} = O(n^\ell)$ sub-tests for an intervention.  For each sub-test, the algorithm's running time is quasi-linear in the sample complexity (Lemma \ref{hellingerTest}), therefore taking a total time of $O({n}^{\ell} K^{\ell(2d+7/4)} n \epsilon^{-2})$.  The number of interventions follow from Lemma~\ref{lemma:counting_argument} and the number of samples follow from the algorithm.  

\paragraph{Correctness.} As in the proof of Theorem~\ref{outTestingResult}, we use Theorem \ref{theorem:mainTheorem} to show that when $\Delta(\mathcal{X},\mathcal{Y}) > \epsilon$, then there exists a subset $\mv{S}$ of some c-component and an $I \in \mv I$ that does not intervene any node in $\mv S$ but intervenes $\Pa(\mv{S})$ with some assignment $\pa(\mv{s})$ such that 
$$ H^2(P_{\mathcal{X}}[\mv{S} \mid do(\pa(\mv{S}))],P_{\mathcal{Y}}[\mv{S} \mid do(\pa(\mv{S}))]) > \epsilon^2 / (2 K^{\ell (d+1)} n).$$
This together with Lemma~\ref{theorem:mainTheorem} proves that $P_{\mathcal{X}}$ and $P_{\mathcal{Y}}$ are far in terms of the total variation distance.  Since the error probability of each sub-test is bounded by at most $1/(6K^{\ell d} 2^{\ell} n)$ and the error probability of $\mv{I}$ being a covering intervention set is at most $1/6$, by union bound, we will have an error of at most $1/3$ over the entire algorithm.
\end{proof}

\subsection{Learning}
Our next result is on learning \smbn s over a known causal graph. Our algorithm is improper, meaning that it does not output a causal model in the form of an \smbn, but rather outputs an oracle which succinctly encodes all the  interventional distributions. See Definition \ref{def:learn} for a rigorous formulation of the problem.

\begin{theorem}[Algorithm for $\clp$] \label{ourLearningResult}
For any given \smcg~$G \in \mathcal{G}_{d,\ell}$ with $n$ vertices and a parameter $\eps > 0$, there exists an algorithm that takes as input an unknown {\smbn} $\mathcal{X}$ over $G$, that performs $O(K^{\ell d} (3d)^\ell \log n )$ interventions to $\mathcal{X}$, taking $\tilde{O}(K^{\ell(2d+3)} n^2 \epsilon^{-4})$ samples per intervention, that runs in time $\tilde{O}\Paren{2^{\ell} K^{\ell (3d+3)} n^3 \epsilon^{-4}}$, and that with probability at least $2/3$, outputs an oracle $\mathcal{N}$ with the following behavior. Given as input any $\mv{T}\subseteq \mv{V}$ and assignment $\mv{t} \in \Sigma^{|\mv{T}|}$, $\mathcal{N}$ outputs an interventional distribution $P_{\mathcal{N}}[\mv{V} \setminus \mv{T}|do(\mv{t})]$ such that:
\[
\delta_{TV}(P_{\mathcal{X}}[\mv{V} \setminus \mv{T} \mid do(\mv{t})], P_{\mathcal{N}}[\mv{V} \setminus \mv{T} \mid do(\mv{t})]) < \epsilon
\] 

When the maximum degree (in-degree plus out-degree) of $G$ is bounded by $d$, then our algorithm uses $O(K^{\ell d} (3d)^\ell \ell d^2 \log K)$ interventions with the same sample complexity and running time as above.
\end{theorem}

\begin{center}
\framebox{
\begin{minipage}[ht]{0.9\linewidth}
\captionof{algocf}{Algorithm for $\clp$}
$\mv I$: Covering intervention set
\begin{enumerate}
\item 
Under each intervention $I \in {\mv I}$:
\begin{enumerate}
\item
Obtain $\tilde{O}(n^2 K^{\ell(2d+3)} \epsilon^{-4})$ samples from the interventional distribution of $I$ in $\mathcal{X}$. 
\item
For each subset $\mv S$ of a c-component, if $I$ does not set $\mv{S}$ but sets $\Pa(\mv{S})$ to $\pa(\mv{S})$, use Lemma \ref{lemma:TVlearn}, Lemma~\ref{lemma:trimUnnecessaryActions} and the obtained samples to learn:
\begin{align*}
P_{\mathcal{N}}[\mv{S} \mid do(\pa(\mv{S}))] \text{ such that } H^2(P_{\mathcal{N}}[\mv{S} \mid do(\pa(\mv{S}))], P_{\mathcal{X}}[\mv{S}  \mid  do(\pa(\mv{S}))]) \leq \dfrac{\epsilon^2}{2 K^{\ell (d+1)} n}
\end{align*}
with probability of error at most $1/(3 K^{\ell d} 2^{\ell} n)$.
\end{enumerate}
\item
Return the following oracle $\mathcal{N}$ that takes as {input}: $\mv{T} \subseteq \mv{V}$ and $\mv{t} \in \Sigma^{|\mv{T}|}$\\[1em]
\framebox{
\begin{minipage}[ht]{0.7\linewidth}
\begin{enumerate}
\item[(i)]Let $C(\mv{V}\setminus\mv{T}) = \{\mv{S}_1, \dots, \mv{S}_p\}$. 
\item[(ii)]Output the distribution $P_{\mathcal{N}}[\mv{V}\setminus\mv{T} \mid do(\mv{t})]$ where for any assignment $\mv{v}\setminus \mv{t}$:
$$P_{\mathcal{N}}[\mv{v}\setminus \mv{t} \mid do(\mv{t})] = \prod_{i=1}^p P_{\mathcal{N}}[\mv{s}_i \mid do(\mv{v}\setminus \mv{s}_i)]$$
\end{enumerate}
\end{minipage}
}
\end{enumerate}
\end{minipage}
}
\end{center}

The covering intervention set used in the algorithm above is as defined in Definition \ref{defn:covering_interventions}.

\medskip
\noindent\textbf{Number of interventions, time, and sample requirements.} 
The number of interventions is obtained using the bound on the size of the covering intervention set from Lemma \ref{lemma:counting_argument}.  When the maximum degree is bounded, we can use Lemma~\ref{lemma:counting_lovaszLocal}.  The number of samples per intervention is obtained from Lemma \ref{lemma:TVlearn}.  Since the algorithm learns at most $n K^{\ell d} 2^{\ell}$ interventions (subroutines), and each subroutine takes time linear in the sample size, the time complexity follows.

\paragraph{Correctness.}  
For any given $\mv{T}$, $do(\mv{t})$, let $C(\mv{V \setminus T}) = \{ \mv{S}_1,\ldots, \mv{S}_p\}$.  Lemma \ref{lemma:c-component-factorization} justifies that $$P_{\mathcal{N}}[\mv{v \setminus t} \mid do(\mathbf{t})] = \prod_i P_{\mathcal{N}}[\mathbf{s}_i  \mid  do(\mv{v} \setminus \mv{s}_i)].$$
Similar to the proof of Theorem \ref{outTestingResult}, using Theorem~\ref{theorem:mainTheorem} and Lemma~\ref{Hellinger-Tv-Inequality}, we get:
\begin{align*}
H^2(P_{\mathcal{N}}[\mv{V} \setminus \mv{T} \mid do(\mv{t})], P_{\mathcal{X}}[\mv{V} \setminus \mv{T} \mid do(\mv{t})]) &< \epsilon^2/2 \\
\implies \delta_{TV} (P_{\mathcal{N}}[\mv{V} \setminus \mv{T} \mid do(\mv{t})], P_{\mathcal{X}}[\mv{V} \setminus \mv{T} \mid do(\mv{t})]) &< \epsilon.
\end{align*}

\section{Main Ingredients of the Analysis}

\subsection{Covering Intervention Sets}
\label{sec:covering-sets}
\begin{lemma}[Counting Lemma: bounded in-degree] \label{lemma:counting_argument}
Let $G \in \mathcal{G}_{d,\ell}$ be a {\smcg} with $n$ vertices and $\Sigma$ be an alphabet set of size $K$.  Then, there is a randomized algorithm that outputs a set $\mv{I}$ of size $O(K^{\ell d} (3d)^\ell (\log n + \ell d \log K + \log(1/\delta)))$.
such that, with probability at least $1-\delta$, $\mv{I}$ is a covering intervention set.
\ignore{
for each c-component $\mv{C}$, for each subset $\mv{S} \subseteq \mv{C}$ and for each assignment $\pa(\mv{S}) \in \Sigma^{|\Pa(\mv{S})|}$, there is a labelling $L \in \mathbf{L}$ that is consistent with $L(\mv{S}) = *^{|\mv{S}|}$ and $L(\Pa(\mv{S})) = \pa(\mv{S})$.
}
\end{lemma}

\begin{proof}
Let $t = K^{\ell d} (3d)^\ell (\log n + 2\ell d \log K + \log(1/\delta))$.  The interventions in $\mv{I}$ are chosen by the following procedure:  For each $j \in [t]$ and for each $V_i \in V$, $V_i$ is observed in $I_j$ with probability $1/(d+1)$ and otherwise, $V_i$ is intervened with the assignment chosen uniformly from $\Sigma$.  Let $V_i = \ast$ denotes that $V_i$ is  not intervened.
Consider a fixed c-component $\mv{C}$, a fixed subset $\mv{S} \subseteq \mv{C}$, a fixed assignment $\pa(\mv{S}) \in \Sigma^{|\Pa(\mv{S})|}$ and a fixed $j \in [t]$.  Now,
\begin{align*}
\Pr[I_j(\mv{S}) = \ast^{|\mv{S}|} \wedge I_j(\Pa(\mv{S})) = \pa(\mv{S})] 
	&= \left(\frac{1}{d+1}\right)^{|\mv{S}|} \cdot \left(\frac{d}{K(d+1)}\right)^{|\Pa(\mv{S})|} \\
&\geq (d+1)^{-\ell} K^{-\ell d} e^{-\ell} \qquad  \text{[Since $|\Pa(\mv{S})| \leq \ell d$ and  $|\mv{S}|\leq \ell$]}\\
&\geq (3d)^{-\ell}K^{-\ell d}.
\end{align*}
This implies that
\begin{align*}
\Pr[\forall j \in [t], (I_j(\mv{S}) \neq \ast^{|\mv{S}|} \vee I_j(\Pa(\mv{S})) \neq \pa(\mv{S}))] 
	\leq \left(1 - (3d)^{-\ell}K^{-\ell d}\right)^t \leq \frac\delta n K^{-2\ell d}. 
\end{align*}
Hence,  
	\begin{align*}
&\quad\Pr[\exists~C\text{-component }\mv{C}, \exists \mv{S} \subseteq C, \exists \pa(\mv{S}) \in \Sigma^{|\Pa(\mv{S})|} , \forall j \in [t], (I_j(\mv{S}) \neq \ast^{|\mv{S}|} \vee I_j(\Pa(\mv{S})) \neq \pa(\mv{S})) ] \\
	&\leq n 2^{\ell}K^{\ell d} \cdot \frac\delta n K^{-2\ell d}  \leq \delta
\end{align*}
by the union bound.
\end{proof}
\begin{remark}
The above proof can be made deterministic by using explicit deterministic constructions of almost $\ell d$-wise independent random variables \cite{AGHP92,EGLNV92}.
\end{remark}

\begin{lemma}[Counting Lemma: bounded total degree] \label{lemma:counting_lovaszLocal}
Let  $G \in \mathcal{G}_{d,\ell}$ be an {\smcg} with $n$ vertices, whose variables take values in $\Sigma$ with $|\Sigma|=K$, and whose maximum degree is bounded by $d$.  Then, there exists covering intervention set $\mv{I}$ of size $O(K^{\ell d} (3d)^\ell \ell d^2 \log K)$.
\end{lemma}
\ignore{
Now let us count the number of possibilities of such $\mv{S}^{\prime},pa(\mv{S}^{\prime})$ (where $\mv{S}^{\prime}$ is a subset of some c-component) that includes at least one vertex from $\mv{S} \cup Pa(\mv{S})$.  For each vertex $V_i$, there are at most $2^{\ell}$ many subsets that include $V_i$ and for each such subset there are at most $K^{ld}$ many possible assignments to the corresponding parents.  Also, for each children of $V_i$, there are at most $2^l K^{ld}$ many possibilities of events that include $V_i$ and $V_i$ has at most $d$ children.  Since there are at most $\ell d +1$ many such $V_i$'s, for a fixed $\mv{C}$,$\mv{S}$ and $pa(\mv{S})$, the number of such $\mv{S}^{\prime},pa(\mv{S}^{\prime})$ that includes at least one vertex from $\mv{S} \cup Pa(\mv{S})$ is at most $(\ell d + 1) (d+1) K^{\ell d}$.  Now using the Lov\'asz local lemma we get, there exists a labelling $L \in \mathbf{L}$ that is consistent with $\mv{S = *}$ and $\Pa(\mv{S}) = \pa(\mv{S})$, for each $\mv{C}_i$, for each $\mv{S} \subseteq \mv{C}_i$ and for each $\pa(\mv{S}) \in \Sigma^{|\Pa(\mv{S})|}$.}

\begin{proof}
Let $t = K^{\ell d} (3d)^\ell (\ell d^2 + \ell d \log K + 2)$.  The interventions in $\mv{I}$ are chosen by the following procedure:  For each $j \in [t]$ and for each $V_i \in V$, $V_i$ is observed in $I_j$ with probability $1/(d+1)$ and otherwise, $V_i$ is intervened with the assignment chosen uniformly from the set $\Sigma$.  Let $V_i = \ast$ denotes that $V_i$ is observed (not intervened).  

For a fixed set $\mv{S}$ that is a subset of a c-component and a fixed assignment  $\pa(\mv{S}) \in \Sigma^{|\Pa(\mv{S})|}$, let $A_{\mv{S}, \pa(\mv{S})}$ be the event: $\forall j \in [t], (I_j(\mv{S}) \neq \ast^{|\mv{S}|} \vee I_j(\Pa(\mv{S})) \neq \pa(\mv{S}))$.
Similar to the proof of Lemma \ref{lemma:counting_argument}, for any fixed $\mv{S}$ and $\pa(\mv{S})$:
$\Pr[A_{\mv{S}, \pa(\mv{S})}] \leq 1/({4 2^{\ell d^2}K^{\ell d}})$.

Now, note that $A_{\mv{S},\pa(\mv{S})}$ and $A_{\mv{T},\pa(\mv{T})}$ are independent if $\Pa(\mv{S})$ and $\Pa(\mv{T})$ are disjoint. For a fixed $\mv{S}$, the number of subsets $\mv{T}$ such that $\Pa(\mv{S}) \cap \Pa(\mv{T}) \neq \emptyset$ is at most $2^{\ell d^2}$ (since, the number of children of the parents of $S$ is at most $\ell d^2$).  Therefore, for a fixed $\mv{S}$ and $\pa(\mv{S})$, $A_{\mv{S}, \pa(\mv{S})}$ is independent of all $A_{\mv{T}, \pa(\mv{T})}$'s except for at most $2^{\ell d^2} K^{\ell d}$ many of them (taking into account the number of possible assignments $\pa(\mv{T})$).
Hence, the Lov\'asz Local Lemma \cite[Chapter 5]{AlonS04} guarantees that there exists a set of $t$ interventions such that $\neg A_{\mv{S}, \pa(\mv{S})}$ for all $\mv{S}$ and $\pa(\mv{S})$.
\end{proof}

\begin{remark}[Explicitness]
Although Lemma \ref{lemma:counting_lovaszLocal} only asserts the existence of a covering intervention, its proof can be turned into a linear time algorithm using the constructive proofs of the Lov\'asz Local Lemma \cite{Mos09, MT10}. 
\end{remark}

\subsection{Subadditivity Theorem for \smbn s}
\label{sec:aub-additivity}
The next theorem states that if two causal models are ``far'', then they must be ``far'' under some ``local'' intervention. 

\begin{theorem} \label{theorem:mainTheorem}
Let $\mathcal{M}$ and $\mathcal{N}$ be two {\smbn}s defined on a {\em known and common} {\smcg} $G \in \mathcal{G}_{d,\ell}$.  
Let $\mv{V}$ be the vertices of $G$.  For a given intervention $do(\mv{t})$, let $\mv{V} \setminus \mv{T}$ partition into $\mcl{C} = \{ \mv{C}_1, \mv{C}_2, \ldots, \mv{C}_{p} \} $, the c-components with respect to the induced graph $G[\mv{V} \setminus \mv{T}]$. 
Suppose
\begin{align}
H^2( P_{\mathcal{M}}[\mv{C}_j  \mid  do(\pa(\mv{C}_j))] , P_{\mathcal{N}}[\mv{C}_j  \mid  do(\pa(\mv{C}_j))] ) \leq \gamma \qquad \forall j \in [p],\forall \pa(\mv{C}_j) \in \Sigma^{|\Pa(\mv{C}_j)|}. \label{hypothesis1}
\end{align}
Then
\begin{align}
H^2 \left( P_{\mathcal{M}}[\mv{V \setminus T}  \mid  do(\mv{t}) ] , P_{\mathcal{N}}[\mv{V \setminus T}  \mid  do(\mv{t}) ] \right) \leq \epsilon \qquad \forall \mv{t} \in \Sigma^{|\mv{T}|}\label{hypothesis2}
\end{align}
where $\epsilon = {\gamma} {|\Sigma|^{\ell (d+1)} n}$.
\end{theorem}

\begin{proof}
Let $\mv{W} = \mv{V \setminus T}= \{ W_1, \ldots, W_r \}$, where the indices are arranged in a topological ordering.   
Here we focus only on distributions on $\mv{W}$ after the intervention $do(\mv{t})$.  That is, our focus is restricted to the graph $G_{\overline{\mv{T}}}$, the intervention $do(\mv{t})$ and the vertices $\mv{W = V \setminus T}$.  We know that
\begin{align}
H^2 \left( P_{\mathcal{M}}[\mv{W} \mid do(\mv{t})],P_{\mathcal{N}}[\mv{W} \mid do(\mv{t})] \right) 
&= 1 - \sum_{\mv{w}} \sqrt{P_{\mathcal{M}}[\mv{w} \mid do(\mv{t})],P_{\mathcal{N}}[\mv{w} \mid do(\mv{t})]}  \nonumber \\
&= 1 - BC \left( P_{\mathcal{M}}[\mv{W} \mid do(\mv{t})],P_{\mathcal{N}}[\mv{W} \mid do(\mv{t})] \right) \label{hellinger-formula}
\end{align} 
where $BC(P_{\mathcal{M}}[\mv{W} \mid do(\mv{t})],P_{\mathcal{N}}[\mv{W} \mid do(\mv{t})])$ is the Bhattacharya coefficient of $P_{\mathcal{M}}[\mv{W} \mid do(\mv{t})]$ and $P_{\mathcal{N}}[\mv{W} \mid do(\mv{t})]$ (see~\eqref{eqn:squaredHellinger}). 

\noindent For each $j \in [p]$, identify the vertices in $\mv{C}_j$ as $\{W_{n_{j,1}}, \dots, W_{n_{j,s_j}}\}$ where $s_j = |\mv{C}_j|$ and $n_{j,1} < \cdots < n_{j, s_j}$.  Using Lemma~\ref{lemma:c-component-factorization}, we express the distributions in terms of the product $\prod_{j=1}^{p} \Pr [\mv{c}_j \mid do(\mv{w} \setminus \mv{c}_j)]$~\cite{tian-pearl}, 
\begin{align*}
 &\quad BC(P_{\mathcal{M}}[\mv{W} \mid do(\mv{t})], P_{\mathcal{N}}[\mv{W} \mid do(\mv{t})])\\
 &= \sum_{\mv{w}} \sqrt{\begin{array}{l} P_{\mathcal{M}}[\mv{w} \mid do(\mv{t})] \\ \newline P_{\mathcal{N}}[\mv{w} \mid do(\mv{t})] \end{array}} \\
 &= \sum_{\mv{w}} \sqrt{ \prod_{j=1}^{p} 
	\begin{array}{l}
	P_{\mathcal{M}}[\mv{c}_j \mid do(\mv{w} \setminus \mv{c}_j)] \\ \newline
	P_{\mathcal{N}}[\mv{c}_j \mid do(\mv{w} \setminus \mv{c}_j)]  
	\end{array}} \\
&= \sum_{\mv{w}} \sqrt{ \prod_{j=1}^{p} \prod_{i=1}^{s_j} 
	\begin{array}{l}
	P_{\mathcal{M}}[ w_{n_{j,i}} \mid w_{n_{j,1}}, \ldots, w_{n_{j,i-1}} , do(\mv{w} \setminus \mv{c}_j) ] \\ \newline
	P_{\mathcal{N}}[ w_{n_{j,i}} \mid w_{n_{j,1}}, \ldots, w_{n_{j,i-1}} , do(\mv{w} \setminus \mv{c}_j) ]
	\end{array}
	} \\
	&= \sum_{\mv{w}} \sqrt{ \prod_{j=1}^{p} \prod_{i=1}^{s_j} 
	\begin{array}{l}
	P_{\mathcal{M}}[ w_{n_{j,i}} \mid w_{n_{j,1}}, \ldots, w_{n_{j,i-1}} , do(\pa( W_{n_{j,1}},\ldots, w_{n_{j,i-1}})) ] \\ \newline
	P_{\mathcal{N}}[ w_{n_{j,i}} \mid w_{n_{j,1}}, \ldots, w_{n_{j,i-1}} , do(\pa( W_{n_{j,1}},\ldots, w_{n_{j,i-1}})) ]
	\end{array}
	} \quad (\text{using Lemma \ref{lemma:independenceLemma}}).
\end{align*}

\noindent For $i \in [s_j]$, let \[
dep(n_{j,i}) := \{ W_{n_{j,1}}, \ldots, W_{n_{j,i}} \} \cup (\Pa(\{ W_{n_{j,1}}, \ldots, W_{n_{j,i}} \}) \setminus \mv{T}).\] For $j \in [p], i \in [s_j]$, let $X_{n_{j,i}}: \Sigma^{|dep(n_{j,i})|} \to [0,1]$ be
\[
X_{n_{j,i}}(\mv{w}_{dep(n_{j,i})}) := \sqrt{
\begin{array}{l}
P_{\mathcal{M}}[w_{n_{j,i}} \mid w_{n_{j,1}}, \ldots, w_{n_{j,i-1}} , do(\pa( W_{n_{j,1}},\ldots, w_{n_{j,i-1}} ))] \\ \newline
P_{\mathcal{N}}[w_{n_{j,i}} \mid w_{n_{j,1}}, \ldots, w_{n_{j,i-1}} , do(\pa( W_{n_{j,1}},\ldots, w_{n_{j,i-1}} ))]
\end{array}
}.
\]
Recall that indices of $\mv W$ follow a topological ordering.  Using this topological ordering and plugging in the expression above, we obtain
\begin{align*}
 &\quad BC(P_{\mathcal{M}}[\mv{W} \mid do(\mv{t})], P_{\mathcal{N}}[\mv{W} \mid do(\mv{t})])= 
\sum_{w_{1}} X_{1}(\mv{w}_{dep(1)}) \sum_{w_{2}} X_{2}(\mv{w}_{dep(2)}) \ldots \sum_{w_{r}} X_{r}(\mv{w}_{dep(r)})
 \end{align*}
where $r = |\mv{W}|$.  In order to prove the theorem, it will suffice to prove that this expression is at least $1-\varepsilon$, whenever~\eqref{hypothesis1} holds. To prove this, we will take the following path, which is essentially an induction on $r$.
For $j \in [p]$, let $b_j = 1$, $dep(\mv{C}_j) = \mv{C}_j \cup (\Pa(\mv{C}_j) \setminus \mv{T})$ and $Y_j(\cdot) = 1$ (a constant function).  Set $\mv{b} = (b_1, \dots, b_p)$, $\mv{dep} := (dep(1), \dots,$  $dep(r),$ $dep(\mv{C}_1),$ $\dots,$ $dep(\mv{C}_p))$, and $\mv{Y}=(Y_1, \dots, Y_p)$. 

In Definition~\ref{programP}, we define an optimization program, $P_{r,p}(\Sigma, \gamma, \mcl{C}, \mv{b}, \mv{dep}, \mv{Y})$ whose objective value is equal to $BC(P_{\mathcal{M}}[\mv{W} \mid do(\mv{t})],P_{\mathcal{N}}[\mv{W} \mid do(\mv{t})])$. In Section~\ref{sec:prog}, we provide the steps to prove a lower bound on the objective of the program, thereby proving a lower bound on $BC(P_{\mathcal{M}}[\mv{W} \mid do(\mv{t})],P_{\mathcal{N}}[\mv{W} \mid do(\mv{t})])$.

\ignore{
\textcolor{orange}{
To check feasibility, note that for each $j \in [p], i \in [s_j]$, and for each assignment $w_{n_{j,1}}, \ldots, w_{n_{j,i-1}},$  $pa( w_{n_{j,i}},\ldots, w_{n_{j,i}} )$,
\begin{align*}
\sum_{w_{n_{j,i}}} X_{n_{j,i}} &=  BC \left( 
\begin{array}{l}
P_{\mathcal{M}}[ W_{n_{j,i}} \mid w_{n_{j,1}}, \ldots, w_{n_{j,i-1}} , do( \pa( W_{n_{j,1}}, \ldots, w_{n_{j,i-1}} ) ) ],\\ \newline
P_{\mathcal{N}}[ W_{n_{j,i}} \mid w_{n_{j,1}}, \ldots, w_{n_{j,i-1}} , do( \pa( W_{n_{j,1}}, \ldots, w_{n_{j,i-1}} ) ) ]  
\end{array}\right)\leq 1
\end{align*}
satisfying \eqref{constraint1}.
SARAVANAN: DO WE NEED THIS THING IN ORANGE, IE, THE FEASIBILITY OR IS IS OBVIOUS? 
}
}

Also, from \eqref{hypothesis1} and \eqref{hellinger-formula}, for all $j \in [p]$ and for all $\mv{w}_{{dep(\mv{C}_j)} \setminus {\mv{C}_j}}$,
\begin{align*}
\sum_{w_{n_{j,1}}} X_{n_{j,1}}(\mv{w}_{dep(n_{j,1})}) \sum_{w_{n_{j,2}}} X_{n_{j,2}}(\mv{w}_{dep(n_{j,2})}) \ldots \sum_{w_{n_{j,s_j}}} X_{n_{j,s_j}}(\mv{w}_{dep(n_{j,s_j})}) \geq 1 - \gamma
\end{align*}
satisfying \eqref{constraint2}. 
Note that $P_{r,p}(\Sigma, \gamma, \mcl{C}, \mathbf{b} , \mv{dep}, \mathbf{Y})$ is a program such that $\max_{j} |dep(C_j)| \leq \ell (d+1) $.  By Lemma \ref{programPMainLemma}, 
\begin{align*}
BC \left( P_{\mathcal{M}}[\mv{W} \mid do(\mv{t})],  P_{\mathcal{N}}[\mv{W} \mid do(\mv{t})] \right) \geq \Opt(P_{r,p}) \geq (1 - |\Sigma|^{\ell (d+1)} \gamma)^{p}. 
\end{align*}
Using this in \eqref{hellinger-formula}, we get
\[
H^2 \left( P_{\mathcal{M}}[\mv{W} \mid do(\mv{t})],  P_{\mathcal{N}}[\mv{W} \mid do(\mv{t})] \right) \leq 1-(1 - |\Sigma|^{\ell (d+1)} \gamma)^{p} \leq p \gamma |\Sigma|^{\ell (d+1)}\leq \epsilon.\qedhere
\]
\end{proof}

\section{Lower Bound on Interventional Complexity}

 Recall that in Section~\ref{sec:main results} we provided non-adaptive algorithms for $\cit$, and $\clp$. In this section we provide lower bounds on the number of interventions that any algorithm must make to solve these problems. Our lower bounds nearly match the upper bounds in Theorem~\ref{outTestingResult}, and Theorem~\ref{ourLearningResult}, even when the algorithm is allowed to be adaptive (namely future interventions are decided based upon the samples observed from the past interventions). In other words, these lower bounds show that adaptivity cannot reduce the interventional complexity.

\begin{theorem}
\label{thm:lb}
There exists a $\smcg$ $G \in \mathcal{G}_{d,\ell}$ with $n$ nodes such that $\Omega(K^{\ell d - 2} \log n)$ interventions are necessary for any algorithm (even adaptive) that solves $\cit$ or $\clp$.   
\end{theorem}

This theorem is proved via the following ingredients. 
\begin{itemize}
\item[] \textbf{Necessary Condition.}
We obtain a necessary condition on the set of interventions ${\mv I}$ of any algorithm that solves $\cit$ or $\clp$.

We will consider {\smcg}s $G$ with a specific structure, and prove the necessary condition for these graphs: The vertices of $G$ are the union of two disjoint sets $\mv A$, and $\mv B$, such that $G$ contains directed edges from $\mv A$ to $\mv B$, and bidirected edges within $\mv B$. Further, all edges in $G$ are one of these two types. The next lemma, proved later in the section, is for graphs with this structure.  

\begin{lemma}\label{lemma:StructureToDistinguish}
Suppose an adaptive algorithm uses a sequence of interventions $\mathbf{I}$ to solve $\cit$ or $\clp$.  Let $\mv{C} \subseteq \mv{B}$ be a c-component of $G$. Then, for any assignment $\pa(\mv{C}) \in \Sigma^{|\Pa(\mv{C})|}$, there is an intervention $I\in \mv{I}$ such that the following conditions hold:
\begin{itemize}
\item [{\bf C1.}] $I$ intervenes $\Pa(\mv{C})$ with the corresponding assignment of $\pa(\mv{C})$,\footnote{In our construction, $\Pa(C)$ always take $\mathbf{0}$ in the natural distribution.  Henceforth, the interventions where some vertices in $\Pa(C)$ are not intervened are not considered here, as they are equivalent to the case when those vertices are intervened with $\mv{0}$.} 
\item [{\bf C2.}] $I$ does not intervene any node in $\mv{C}$. 
\end{itemize}
\end{lemma} 
\item[] \textbf{Existence.}
We then show that there is a graph with the structure mentioned above for which $\mv {I}$ must be $\Omega( K^{\ell d - 2} \log n)$ in order for the condition to be satisfied. More precisely, 
\begin{lemma}\label{lemma:lb_construction}
There exists a $G$, and a constant $c$ such that for any set of interventions $\mathbf{I}$ with $|\mathbf{I}|< c\cdot K^{\ell d - 2} \log n$, there is a $\mv{C} \subseteq \mv{B}$, which is a c-component of $G$, and an assignment $\pa(\mv{C})$ such that no intervention in $\mathbf{I}$ 
\begin{itemize}
\item assigns $\pa(\mv{C})$ to $\Pa(\mv{C})$, and
\item observes all variables in $\mv{C}$.
\end{itemize}
\end{lemma}
\end{itemize}
\begin{proof}[Proof of Lemma~\ref{lemma:lb_construction}]
We show existence of such a $G$ using a probabilistic argument. We consider $\mv A =\mv A_r\cup \mv A_f$, where $\mv{A}_r:=\{A_1,\ldots,A_{n}\}$, and $\mv{A}_f := \{ A_{n+1},\ldots,A_{n+(\ell d)-2} \}$. We consider $\mv{B}:=\mv{B}_1 \cup \mv{B}_2 \cup \ldots \cup \mv{B}_{n/\ell}$, where for each $i \in [n/\ell]$, $\mv{B}_{i} = \{ B_{i,1},B_{i,2},\ldots,B_{i,\ell} \}$.  $\mv{V}=\mv A\cup \mv B$ will be the set of observable nodes in the graph. Therefore, the number of nodes is $|\mv V|= 2n+\ell d-2 = O(n)$. 

 The set of unobservable nodes are such that the following is satisfied:
\begin{itemize}
\item  $\mv{B}_i$ is a c-component in $G$, for each $\mv B_i$. 
\end{itemize}

We consider random directed bipartite graphs on $\mv{V}$ generated as follows, where all the edges go from $\mv A$ to $\mv B$. Each c-component $\mv{B}_i$ has exactly $\ell d$ parents, chosen as follows:
\begin{itemize}
\item 
$\mv{A}_f \subset \Pa(\mv{B}_i)$, namely every vertex of $\mv A_f$ is the parent of at least one node in $\mv B_i$. 
\item
The remaining two parents of $\mv{B}_i$ are chosen randomly from $\mv{A}_{r}$ with edge density $p:=2/n$.
\end{itemize}

Let $\mathbf{I}$ be a set of interventions that satisfies the conditions of  Lemma~\ref{lemma:StructureToDistinguish}. 
 Let $\mv I'\subseteq\mv I$ be the interventions that intervene \emph{all the nodes} in $\mv A$. The nodes in $\mv A_f$ can be intervened in $|\Sigma|^{|\mv A_f|}= K^{\ell d - 2}$ ways. This induces a partition of $\mathbf{I'}$ into $K^{\ell d - 2}$ parts, where the interventions in each partition intervenes $\mv A_f$ with the same assignment. Let $\{\mv I_1, \ldots , \mv I_{K^{\ell d-2}}\}$ such that $\mv I' = \mv I_1\cup \ldots\cup \mv I_{K^{\ell d-2}}$ be this partition. We will show that for each $j$, $|\mv I_{j}|=\Omega(\log n)$, implying that 
\[
|\mv I|\ge |\mv I'| \ge K^{\ell d - 2}\cdot \Omega(\log n) = \Omega(K^{\ell d - 2}\log n).
\]

Consider a $\mv I_j$, with $|\mv I_j|=t$. Further, for simplicity we assume that $K=2$ for this part, and that $\Sigma = \{0,1\}$. Since all the nodes in $\mv A_r$ are intervened, consider one such node. For any node in $\mv A_r$ consider the $t$ bit binary string denoting whether it is intervened with 0 or 1 in the $t$ interventions. This divides the set $\mv{A}_r$ into $2^{t}$ cells $\mv Z_1, \ldots, \mv Z_{2^{t}}$, where two nodes are in the same cell if they are intervened with the identical value by each intervention in $\mv I_j$. The expected number of pairs of vertices in $\mv Z_h$ that are both parents of some vertex in $\mv{B}$ is $O(p |\mv{Z}_h|^2)$.  Therefore, the expected number of pairs of vertices that are both parents of some vertex in $\mv{B}$ and also belong to the same cell is $O\left(\sum_{h} p |\mv{Z}_h|^2\right)$, which is at least $O(p  n^2  2^{-{t}})$ (since $\sum_{h} \mv{Z}_h = n$).  Now for any such pair of vertices $A,A^{\prime} \in \mv{A}_{r}$ that belong to the same cell, there exists no intervention such that $A = 0$ and $A^{\prime} = 1$, contradicting to our requirement.  Therefore, $p n^2  2^{-{t}} < 1$ which implies $t$ is at least $\Omega(\log n)$.
\end{proof} 

Combining these two lemmas, we obtain the lower bound for the adaptive versions of $\cit$ and $\clp$.  Now we proceed to prove Lemma~\ref{lemma:StructureToDistinguish}. 

\begin{proof}[Proof of Lemma \ref{lemma:StructureToDistinguish}] 
 
In our construction we consider models where $\mv{A}$ is assigned $\mv{0}^{|\mv{A}|}$ with probability one in the observable distribution. In other words, each $A_i\in \mv A$ takes value 0 with probability one. Consider any intervention $I$ that targets a $\mv A'\subseteq\mv A$. Consider the intervention $I'$ that intervenes $\mv A'$ the same way as $I$, but intervenes the nodes in $\mv A\setminus \mv A'$ with 0's. Since there are no incoming arrows to $\mv A$, the distribution of $\mv I'$ will be the same as $\mv I$. Therefore, we assume that each intervention $I$ we make intervenes \emph{all the vertices} in $\mv A$. 

Suppose there is an algorithm that makes a series of interventions $\mv{I}$ that do not satisfy the conditions of Lemma~\ref{lemma:StructureToDistinguish}. In other words, there exists a c-component $\mv{C} \subseteq \mv{B}$ and an assignment $\pa(\mv{C})$, such that no intervention in $\mv{I}$ satisfies {\bf C1} and {\bf C2}.  Let $\mv{C} = \{ V_1, V_2, \ldots, V_{\ell} \}$ and $\Pa(\mv{C}) = \{ W_1,W_2, \ldots, W_{s} \}$.


Let $G^{\prime}$ be a subgraph of $G$ on the vertices $\mv{C} \cup \Pa(\mv{C})$ whose edge set satisfies the following: 
\begin{itemize}
\item $\mv{C}$ contains exactly $\ell - 1$ bidirected edges that form a tree.
\item each of the parent vertices $W_i$ has exactly one child node in $\mv{C}$. 
\end{itemize}



In our construction, we consider models where the distribution on the rest of the vertices of $G$ (\textit{i.e.,} $\mv{V} \setminus (\mv{C \cup \Pa(C)})$) will be independent of the distribution on $\mv{C \cup \Pa(C)}$. 
Therefore, we can restrict our focus on $G^{\prime}$.  We will show the existence of two models $\mathcal{M}$ and $\mathcal{N}$ on $G^{\prime}$ such that: 


\begin{enumerate}[label=\textbf{S.\arabic*},ref=S.\arabic*]
\item Let $\mv{T} \subseteq (\mv{C} \cup \Pa(\mv{C}))$, $\mv{t} \in \Sigma^{|\mv{T}|}$.  Let $\{\mv{C}_1,\ldots,\mv{C}_q\}$ be the c-components of the induced graph $G^{\prime}[\mv{C} \setminus \mv{T}]$.  Suppose under the intervention $do(\mv{t})$, the conditions {\bf C1}, or {\bf C2} is not satisfied, then, the distributions over ${\mv{C} \setminus \mv {T}}$ in $\mathcal{M}$ and $\mathcal{N}$ are identical under $do(\mv{t})$, namely, 
\ignore{
\begin{align*}
P_{\mathcal{M}}[\mv{C \setminus T}  \mid  do(\mv{t})] = P_{\mathcal{N}}[\mv{C \setminus T}  \mid  do(\mv{t})].
\end{align*}
Suppose, the intervention $do(\mv{t})$ does not satisfy the two conditions of the Lemma statement.  Then, }
\begin{align*}
P_{\mathcal{M}}[\mv{C \setminus T}  \mid  do(\mv{t})] = \prod_{i} P_{\mathcal{M}}[\mv{C}_i  \mid  do(\mv{t})] = \prod_{i} P_{\mathcal{N}}[\mv{C}_i  \mid  do(\mv{t})] = P_{\mathcal{N}}[\mv{C \setminus T}  \mid  do(\mv{t})]
\end{align*}
where for each $i$, $P_{\mathcal{M}}[\mv{C}_i  \mid  do(\mv{t})] = P_{\mathcal{N}}[\mv{C}_i  \mid  do(\mv{t})]$ and is a uniform distribution over $\{0,1\}^{|\mv{C}_i|}$,   \label{itemtwo}
\item $\delta_{TV}(P_{\mathcal{M}}[\mv{C}  \mid  do(\pa(\mv{C}))], P_{\mathcal{N}}[\mv{C}  \mid  do(\pa(\mv{C}))])=1$.\footnote{Recall that $\pa(\mv{C})$ is the assignment that gets fixed after the algorithm fixes the sequence $\mv{I}$.} \label{itemthree} 
\end{enumerate} 

Recall that the sequence of interventions performed by an (adaptive) algorithm is denoted by $\mv{I}$.  The assignment $\pa(\mv{C})$ gets fixed only after the algorithm fixes {\em all} the interventions in $\mv{I}$.  However, we know that any intervention in $\mv{I}$ belongs to the category \ref{itemtwo}.  And for each such intervention in $\mv{I}$, the corresponding distributions on models $\mathcal{M}$ and $\mathcal{N}$ are equal, and is defined by a set of uniform distributions over the c-components.  Therefore, we can construct an adversary that, for each intervention in $\mv{I}$ performed by the algorithm (sequentially), outputs a distribution\footnote{We consider the worst case, where the algorithm is provided with infinite samples.} based on \ref{itemtwo}.  When the algorithm terminates, the assignment $\pa(\mv{C})$ gets fixed, and we can show the existence of two models $\mathcal{M}$ and $\mathcal{N}$ such that 
\begin{itemize}
\item the models agree on all the interventional distributions in $\mv{I}$, and all such distributions also match the corresponding distributions that were revealed by the adversary.
\item $\delta_{TV}(P_{\mathcal{M}}[\mv{C}  \mid  do(\pa(\mv{C}))], P_{\mathcal{N}}[\mv{C}  \mid  do(\pa(\mv{C}))])=1$. 
\end{itemize}


Moreover, we can construct such an adversary that outputs distributions in the same way, for all the c-components $\mv{C} \subseteq \mv{B}$ of $G$.  Thus, an explicit construction of two models $\mathcal{M}$ and $\mathcal{N}$ on $G^{\prime}$ that generates distributions according to \ref{itemtwo} and~\ref{itemthree} would conclude our proof.  The remainder of the proof is dedicated towards this goal.

Let $\mv{U}$ be the set of all unobservable variables in $G^{\prime}$.  Let $\mv{U}^{V_i} \subseteq \mv{U}$ represent the bidirected edges incident to $V_i$ in $G^{\prime}$.  Also, for each variable $V_i$ we have an additional boolean random variable $R_i$ that provides randomness to $V_i$. All the randomness in the models $\mathcal{M}$ and $\mathcal{N}$ we construct are in the hidden variables $U_i$'s and the $R_i$'s. In other words, the observable variables are a deterministic function of these. 
The models $\mathcal{M}$ and $\mathcal{N}$ are defined as follows:  
\begin{enumerate}
\item \begin{enumerate}
\item For each bidirected edge $U_i \in \mv{U}$, $U_i$ is a $\textrm{Bern(0.5)}$ random variable in both $\mathcal M$, and $\mathcal N$. 
\item In each model, $R_i$'s  are also independent $\textrm{Bern(0.5)}$ random variables. 
\end{enumerate}
\item For each $i \in [s]$, $W_i = 0$ with probability one in both $\mathcal M$, and $\mathcal N$.
\item For each $V_i \in \mv{C}$, with probability one: 
\begin{enumerate}
\item when $\Pa(V_i)$ is not consistent with $\pa(V_i)$, then $V_i=\mathsf{XOR}(\mv{U}^{V_i},R_i)$ in both both $\mathcal M$, and $\mathcal N$. 
\item when $\Pa(V_i)$ is consistent with $\pa(V_i)$ and $i \neq 1$, then $V_i = \mathsf{XOR}(\mv{U}^{V_i})$ in both $\mathcal M$, and $\mathcal N$.
\item when $\Pa(V_i)$ is consistent with $\pa(V_i)$ and $i = 1$, $V_i$ takes 
\begin{itemize}
\item $V_i = \mathsf{XOR}(\mv{U}^{V_i})$ in $\mathcal{M}$, and  
\item $V_i = \mathsf{XNOR}(\mv{U}^{V_i})$ in $\mathcal{N}$.
\end{itemize}
\end{enumerate}
\end{enumerate}

\textbf{Case 1: When $I$ respects \ref{itemtwo}.} Consider an intervention $I$, identified by $do(\mv{t})$,  that respects \ref{itemtwo}.  That is, either $I$ intervenes some node in $\mv{C}$, or $I$ does not intervene $\Pa(\mv{C})$ with the assignment $\pa(\mv{C})$.  Let $\{\mv{C}_1,\ldots,\mv{C}_q\}$ be the c-components of the graph induced by $\mv{C} \setminus \mv{T}$.  Note that the models $M$, and $N$, differ only on the function $V_1$.  Therefore, when $V_1$ is intervened in $I$, it is easy to see that the required distributions are equal, and is a product of uniform distributions over the c-components.\footnote{Recall that our objective is to prove: $P_{\mathcal{M}}[\mv{C \setminus T}  \mid  do(\mv{t})] = \prod_{i} P_{\mathcal{M}}[\mv{C}_i  \mid  do(\mv{t})] = \prod_{i} P_{\mathcal{N}}[\mv{C}_i  \mid  do(\mv{t})] = P_{\mathcal{N}}[\mv{C \setminus T}  \mid  do(\mv{t})]$, where for each $i$, $P_{\mathcal{M}}[\mv{C}_i  \mid  do(\mv{t})]=P_{\mathcal{N}}[\mv{C}_i  \mid  do(\mv{t})]$ is a uniform distribution over $\{0,1\}^{|\mv{C}_i|}$.}  Suppose $V_1$ is not intervened in $I$, and without loss of generality let $\mv{C}_1$ be the c-component that contains $V_1$.  Since the models differ only on $V_1$, it is easy to see that $P_{\mathcal{M}}[\mv{C}_i  \mid do(\mv{t})]= P_{\mathcal{N}}[\mv{C}_i  \mid  do(\mv{t})]$ for all $i \neq 1$, and is uniform over $\{0,1\}^{|\mv{C}_i|}$.  Hence, it is sufficient to prove that $P[\mv{C}_1  \mid  do(\mv{t})]$'s are equal and uniform in both models.  Let $\mv{S}$ be the set of all $U_i$'s and $R_j$'s of the following type: a) $U_i$'s that have one child in $\mv{C}_1$ and another child in $\mv{T}$; b) $R_j$'s with respect to $V_j\in \mv{C}_1$ such that $\pa(V_j)$\footnote{We refer $\pa(V_j)$ with respect to the assignment $\pa(\mv{C})$.} is inconsistent with $\mv{t}$ (\textit{i.e.,} $V_j \in \mv{C}_1$ that computes $\mathsf{XOR}(\mv{U}^{V_j},R_j)$).  Now, for any fixed assignment $\mv{a}_{\mv{{C}^{-1}_{1}}}$ to $\mv{C}_1 \setminus \{V_1\}$ and $\mv{a}_{\mv{S}}$ to $\mv{S}$, because the bi-directed edges within $\mv{C}_1$ form a `{\em tree}', the value of every unobservable variable within $\mv{C}_1$\footnote{We refer to the unobservable variables $U_i$'s where both the children of $U_i$ lie in $\mv{C}_1$.} can be computed.  Note that $V_1$ computes $\mathsf{XOR}(\mv{a}_{\mv{{C}^{-1}_{1}}},\mv{a}_{S})$ in $\mathcal{M}$, and $\mathsf{XNOR}(\mv{a}_{\mv{{C}^{-1}_{1}}},\mv{a}_{S})$ in $\mathcal{N}$.  However, we know that $\mv{S}$ is a non-empty set and the bit parities of $\mv{S}$ are uniformly distributed in both the models.  This implies $P_{\mathcal{M}}[\mv{C}_1  \mid do(\mv{t})]= P_{\mathcal{N}}[\mv{C}_1  \mid  do(\mv{t})]$, and is a uniform distribution over $\{0,1\}^{|\mv{C}_1|}$.              

\textbf{Case 2: When $I$ respects \ref{itemthree}.} Consider an intervention $I$ that respects \ref{itemthree}.  That is, $\Pa(\mv{C})$ is intervened with the assignment $\pa(C)$ in $I$, and no node of $\mv{C}$ is intervened in $I$.  Consider the set of variables $\mv{S}$ as defined before for the \ref{itemtwo} case.  Note that $\mv{S}$ is empty here.  This implies, for any fixed assignment $\mv{a}_{\mv{{C}^{-1}}}$ to $\mv{C} \setminus \{V_1\}$, $V_1$ computes $\mathsf{XOR}(\mv{a}_{\mv{{C}^{-1}}})$ in $\mathcal{M}$, and $V_1$ computes $\mathsf{XNOR}(\mv{a}_{\mv{{C}^{-1}}})$ in $\mathcal{N}$.  This implies, the supports of $P_{\mathcal{M}}[\mv{C} \mid do(\pa(\mv{C}))]$ and $P_{\mathcal{N}}[\mv{C} \mid do(\pa(\mv{C}))]$ are disjoint, and therefore the total variation distance is $1$.

Hence, irrespective of the number of samples taken from the interventions of $\mathbf{I}$, any adaptive algorithm that solves $\cit$ or $\clp$ must consider a sequence of interventions that satisfies the conditions {\bf C1} and {\bf C2}.

\ignore{
In Lemma~\ref{lemma:disparitylemma}, we show that when the parents of $\mv{C}$ are targeted for intervention with the assignment $\pa(\mv{C})$, the supports of $P_{\mathcal{M}}[\mv{C} \mid do(\pa(\mv{C}))]$ and $P_{\mathcal{N}}[\mv{C} \mid do(\pa(\mv{C}))]$ are disjoint.  This implies the total variation distance between $P_{\mathcal{M}}[\mv{V} \setminus \mv{T} \mid do(\mv{t})]$ and $P_{\mathcal{N}}[\mv{V} \setminus \mv{T} \mid do(\mv{t})]$ is $1$ when $\mv{t} = \pa(\mv{C})$.  Therefore, $\mathcal{M}$ and $\mathcal{N}$ respect \ref{itemthree}.  

Also, we know that when the nodes in $\Pa(\mv{C})$ are intervened with an assignment that is inconsistent with $\pa(\mv{C})$, or when some node in $\mv{C}$ is intervened, (\textit{i.e.,} when the intervention belongs to $\ref{itemtwo}$) the corresponding interventional distributions with respect to models $\mathcal{M}$ and $\mathcal{N}$ are equal (shown in Lemma~\ref{lemma:disparitylemma}).  Hence, irrespective of the number of samples taken from the interventions of $\mathbf{I}$, any adaptive algorithm that solves $\cit$ or $\clp$ must consider a sequence of interventions that satisfies the required conditions ({\bf C1} and {\bf C2}).
}
\end{proof}

\section{Program $P_{r,p}$ and Properties}\label{sec:prog}

In this section, we gather the technical tools used to prove the subadditivity result, Theorem \ref{theorem:mainTheorem}. We formulate our claims at a higher level of abstraction than needed for our purposes, so that the essence of the argument becomes clearer.

We begin by defining the optimization problem, and then describe it at a high level. 
\begin{definition}[\textbf{Program} $P_{r,p}(\Sigma,\gamma,\mathcal{C},\mathbf{b},\mathbf{dep},\mathbf{Y})$]  \label{programP}
For integers $r,p\geq 0$, suppose the following are given:
\begin{enumerate}
\item
 an alphabet set $\Sigma$,
\item
 $\gamma \in (0,1)$,
\item
a partition\footnote{Here, we allow some members of $\mcl{C}$ to be empty sets.} $\mathcal{C}$ of $[r]$ into $\mv{C}_1, \mv{C}_2,\ldots,\mv{C}_p$, where for each $j \in [p]$, $s_j = |\mv{C}_j|$ and the elements of $\mv{C}_j$ are $\{n_{j,1}, \dots, n_{j, s_j}\}$ in increasing order,
\item
a vector $\mathbf{b}=(b_1,b_2,\ldots,b_p) \in {[0,1]}^p$,
\item
a vector of sets  $\mathbf{dep}=(dep(1),\ldots,dep(r),dep(\mv{C}_1),\ldots,dep(\mv{C}_p))$ such that:
\begin{align*}
[n_{j,i}] \supseteq dep(n_{j,i}) &\supseteq \{ n_{j,i} \} \cup dep(n_{j,i-1}) &\forall j \in [p], i \in [s_j] \\
s_j \neq 0 &\implies dep(\mv{C}_j) \supseteq dep(n_{j, s_j}) &\qquad \forall j \in [p]\\
s_j = 0 &\implies dep(\mv{C}_j)=\emptyset &\qquad \forall j \in [p]
\end{align*}
\item
a set of functions $\mathbf{Y} = (Y_1, Y_2, \ldots ,Y_p)$, where $Y_j\colon \Sigma^{|dep(\mv{C}_j)|} \to [0,1]$.
\end{enumerate}
The program $P_{r,p}(\Sigma,\gamma,\mcl{C},\mathbf{b},\mv{dep},\mathbf{Y})$ is  the following optimization problem over $\mv{X} = (X_1, \dots, X_r)$ where $X_i \colon \Sigma^{|dep(i)|} \to [0,1]$:
\begin{alignat}{2}
& \underset{\mv{X}}{\min} f_{r,p}(\mv{X}) \defeq
\sum_{a_1 \in \Sigma} X_1(\mv{a}_{dep(1)}) \sum_{a_2 \in \Sigma} X_2(\mv{a}_{dep(2)}) &&\cdots \sum_{a_r \in \Sigma} X_r(\mv{a}_{dep(r)})\cdot \prod_{j=1}^{p} Y_j(\mv{a}_{dep(\mv{C}_j)})  \nonumber \\
& \text{subject to} \nonumber \\
& \sum_{a_i \in \Sigma} X_i(\mv{a}_{dep(i)}) \leq 1 & & \qquad \quad \forall i \in [r], \forall \mv{a}_{dep(i) \setminus \{i\}} \in \Sigma^{|dep(i)\setminus \{i\}|}  \label{constraint1} \\
&
	\sum\limits_{a_{n_{j,1}} \in \Sigma} X_{n_{j,1}}(\mv{a}_{dep(n_{j,1})})  \sum\limits_{a_{n_{j,2}} \in \Sigma} X_{n_{j,2}}(\mv{a}_{dep(n_{j,2})}) \cdots &&\sum\limits_{a_{n_{j, s_j}} \in \Sigma} X_{n_{j,s_j}} (\mv{a}_{dep(n_{j,s_j})}) \cdot Y_j(\mv{a}_{dep(\mv{C}_j)}) \nonumber \\
	& && \geq 
	1-b_j \gamma \qquad \forall j \in [p], \forall \mv{a}_{dep(\mv{C}_j) \setminus \mv{C}_j}  \label{constraint2}
\end{alignat}
\end{definition}

The variables of this program are functions that are based on the Bhattacharya coefficients between distributions\footnote{The distributions may be interventional, or conditional, or a combination of condional and interventional distributions.} on certain variables, and were described in Section~\ref{sec:aub-additivity}.~\eqref{constraint1} captures the fact that the Bhattacharyya coefficient is at most one.  ~\eqref{constraint2} captures the closeness constraint in Theorem~\ref{theorem:mainTheorem}, \textit{i.e.,}~\eqref{hypothesis1}.
Proving a lower bound on the objective value of this program will suffice to prove Theorem~\ref{theorem:mainTheorem}. The remainder of this section is dedicated towards this goal.
Let $\Opt(P_{r,p})$ denote the optimal value of the program.  The next three lemmas (Lemmas \ref{lemma:dependent-set-reduction}, \ref{lemma:Y-Rreduction} and \ref{lemma:R-Elimination}) all have the following flavor:
\begin{itemize}
	\item They take as input an optimization problem (program $P_{r,p}$), and output a new program $P_{r',p}^{\mt{new}}$. 
	\item The optimal value of the program only goes down.
	\item The new program is \emph{simpler} to analyze.\footnote{We understand that this item is very subjective.}
\end{itemize} 
We pass the original program $P_{r,p}$ through the first lemma, and pass its output through the second.  The second lemma is applied multiple times until the output program satisfies a particular property.  The obtained program is then passed through the third lemma to obtain a new program $P_{r-1,p}$ (with a reduced value of $r$), and the steps repeat.  The above procedure reduces to a program with $r=0$, namely to a program of the form $P_{0,p}$.  We can lower bound the objective of this program by simply using~\eqref{constraint2}. Combining these will yield a  lower bound on the optimum of the original program $P_{r,p}$, thus proving Theorem~\ref{theorem:mainTheorem}.

The first lemma takes a program as input and outputs a new program with a smaller optimal value that satisfies $dep(r) = dep(\mv{C}_{f})$ (where $r \in \mv{C}_{f}$).

\begin{lemma}[Dependent Set Reduction] \label{lemma:dependent-set-reduction}
Suppose $r \in \mv{C}_f$. Let $P_{r,p}^{\mt{new}}$ be the program obtained from $P_{r,p}$ by replacing $dep(r)$ by $dep(\mv{C}_f)$, then
\begin{align*}
	\Opt(P_{r,p}) \geq \Opt(P_{r,p}^{\mt{new}}).
\end{align*}
\end{lemma}

\begin{proof}
Our goal is to reduce the given program $P_{r,p}$ to a different program $P_{r,p}^{\mt{new}}$ such that $\Opt(P_{r,p}) \geq \Opt(P_{r,p}^{\mt{new}})$, where $P_{r,p}^{\mt{new}}$ is defined from $P_{r,p}$ by defining $dep(r)$ to be $dep(\mv{C}_f)$.  

Let $\mv{X}^{\mt{old}} = \{X_1^{\mt{old}}, \ldots X_r^{\mt{old}} \}$ be an optimal solution of $P_{r,p}$.  Now we construct a \emph{feasible} solution $\mv{X}^{\mt{new}} = \{X_1^{\mt{new}}, \ldots, X_r^{\mt{new}}\}$ for the program $P_{r,p}^{\mt{new}}$, such that $f_{r,p}^{\mt{new}}(\mv{X}^{\mt{new}})=f_{r,p}(\mv{X}^{\mt{old}})=\Opt(P_{r,p})$.  For all $i \neq r$, we define $X_i^{\mt{new}} = X_i^{\mt{old}}$. For $i=r$, we define $X^{\text{new}}_r(\mv{a}_{dep^{\mt{new}}(r)}) = X^{\text{old}}_r(x_{dep(r)})$.  In other words, $X^{\mt{new}}_r$ ignores the new variables added to $dep^{\mt{new}}(r)$.  Therefore,
\begin{align*}
f_{r,p}^{\mt{new}}(\mv{X}^{\mt{new}}) 
	&= \sum_{a_1 \in \Sigma} X^{\mt{new}}_1(\mv{a}_{dep(1)}) \cdots  \sum_{a_r \in \Sigma} X^{\mt{new}}_r(\mv{a}_{dep^{\mt{new}}(r)})\cdot \prod_{j=1}^{p} Y_j(\mv{a}_{dep(\mv{C}_j)}) \\
\ignore{	&= \sum_{a_1 \in \Sigma} X^{\mt{new}}_1(\mv{a}_{dep(1)}) \cdots  \sum_{a_r \in \Sigma} X^{\mt{new}}_r(\mv{a}_{dep(\mv{C}_f)})\cdot \prod_{j=1}^{p} Y_j(\mv{a}_{dep(\mv{C}_j)}) \quad (\because dep^{\mt{new}} = dep(\mv{C}_j)) \\}
	&= \sum_{a_1 \in \Sigma} X^{\mt{old}}_1(\mv{a}_{dep(1)}) \cdots \sum_{a_r \in \Sigma} X^{\mt{old}}_r(\mv{a}_{dep(r)})\cdot \prod_{j=1}^{p} Y_j(\mv{a}_{dep(\mv{C}_j)}) \quad \text{(by the definition of $\mv{X}^{\mt{new}}$)}\\
	&= f_{r,p}(X^{\mt{old}}) \quad \quad \text{(by the definition of $f_{r,p}$)}.
\end{align*}

For the program $P_{r,p}^{\mt{new}}$, when $i \neq r$, $\mv{X}^{\mt{new}}$ satisfies the constraints in \eqref{constraint1} (since the functions $X^{\mt{old}}_{i}$ and $X^{\mt{new}}_{i}$ are the same).  Similarly, for $j \neq f$, constraints in \eqref{constraint2} of the program $P_{r,p}^{\mt{new}}$ are valid.  When $i=r$ in \eqref{constraint1}, for each $\mv{a}_{dep^{\mt{new}}(r) \setminus \{r\}}$, we get 
\begin{align*}
\sum_{a_{r}} X_r^{\mt{new}}(\mv{a}_{dep^{\mt{new}}(r)})
& = \sum_{\mv{a}_{r}} X^{\mt{old}}_{r}(\mv{a}_{dep(r)})\leq 1. 
\end{align*}

When $j=f$ in \eqref{constraint2}, for all $\mv{a}_{dep(\mv{C}_j) \setminus \mv{C}_j}$, since $n_{f,s_f} = r$ we get,
\begin{align*}
&\sum\limits_{a_{n_{f,1}} \in \Sigma} X^{\mt{new}}_{n_{f,1}}(\mv{a}_{dep(n_{f,1})}) \cdots  \sum\limits_{a_{r} \in \Sigma} X^{\mt{new}}_{{r}} (\mv{a}_{dep^{\mt{new}}({r})}) \cdot Y_f(\mv{a}_{dep(\mv{C}_f)}) \\
	& = \sum\limits_{a_{n_{f,1}} \in \Sigma} X^{\mt{old}}_{n_{f,1}}(\mv{a}_{dep(n_{f,1})}) \cdots   \sum\limits_{a_{r} \in \Sigma} X^{\mt{old}}_{{r}} (\mv{a}_{dep({r})}) \cdot Y_f(\mv{a}_{dep(\mv{C}_f)}) 
	 \quad \text{(from definition of $\mv{X}^{\mt{new}}$)}\\
	 &\geq 
	1- b_f \gamma \qquad \text{(using \eqref{constraint2}).}
\end{align*}
This implies $X^{\mt{new}}$ is a feasible solution for $P_{r,p}^{\mt{new}}$ and hence $\Opt(P_{r,p}) \geq \Opt(P_{r,p}^{\mt{new}})$.  
\end{proof}

The next lemma takes a program $P_{r,p}$ as input and outputs a new program (with a smaller optimal value) that satisfies $r \notin dep(\mv{C}_h)$ (for some given $\mv{C}_h$ such that $r \notin \mv{C}_h$). 

\begin{lemma}[Y-R Reduction] \label{lemma:Y-Rreduction}
Let $P_{r,p}(\Sigma,\gamma,\mcl{C},\mathbf{b},\mathbf{dep},\mathbf{Y})$ be a given program, and there exists $h \in [p]$ such that $r \notin \mv{C}_{h}$ and $r \in dep(\mv{C}_h)$.  Then, there exists a program $P^{\mt{new}}_{r,p}\paren{\Sigma,\gamma,\mcl{C},\mathbf{b}^{\mt{new}},\mathbf{dep}^{\mt{new}}\linebreak[0],\mathbf{Y}^{\mt{new}}}$ such that 
\begin{align*}
	\Opt(P_{r,p}) \geq \Opt(P^{\mt{new}}_{r,p}),
\end{align*}
where
\begin{enumerate}
\item \makebox[5cm][t]{$b^{\mt{new}}_h = |\Sigma| \cdot b_h$}
\item \makebox[5cm][t]{$dep^{\mt{new}}(\mv{C}_h) = dep(\mv{C}_h) \setminus \{ r \}$} 
\item \makebox[5cm][t]{$b^{\mt{new}}_j = b_j$} $\forall j \in [p] \setminus \{ h \}$ 
\item \makebox[5cm][t]{$dep^{\mt{new}}(C_j) = dep(C_j)$} $\forall j \in [p] \setminus \{h\}$
\item \makebox[5cm][t]{$dep^{\mt{new}}(i) = dep(i)$ } $\forall i \in [r] $
\item \makebox[5cm][t]{$Y^{\mt{new}}_j (\mv{a}_{dep(\mv{C}_j)}) = Y_j (\mv{a}_{dep(\mv{C}_j)})$} $\forall j \in [p] \setminus \{h\}, \forall \mv{a}_{dep(\mv{C}_j)}$.
\end{enumerate}
\end{lemma}

\begin{proof}
Let $\mv{X}^{\prime}$ be an optimal solution of $P_{r,p}$.  Note that, since $dep(\mv{C}^{\mt{new}}_h) = dep(\mv{C}_h) \setminus \{ r \}$, our goal is to find a function $Y^{\mt{new}}_h : \Sigma^{|dep(\mv{C}^{\mt{new}}_h)|} \to [0,1]$, whose domain size is smaller than the domain size of $Y_h$ (as $Y_h^{\mt{new}}$ is independent of the value of $a_r$), that satisfies the required constraints.  \paragraph{} 
For a given set of functions $\mv{X}$, a subset $\mv{S} \subseteq [r]$,  and for a given assignment $a_{\mv{S}}$ to $\mv{S}$, let $f_{r,p}( \mv{X} ) |_{a_{\mv{S}}}$ represent the sum of all terms in $f_{r,p}(\mv{X})$ that are consistent with the assignment $a_{\mv{S}}$.  Note that 
\begin{align}
f_{r,p}(\mv{X}^{\prime}) = \sum\limits_{\mv{a}_{dep(\mv{C}_h)}}  f_{r,p}(\mv{X}^{\prime})|_{\mv{a}_{dep(\mv{C}_h)}}.  \label{summation-equality}
\end{align}

For each $\mv{a}_{dep(\mv{C}_h) \setminus \{ r \}}$, let
\begin{enumerate}
\item $z_h(\mv{a}_{dep(\mv{C}_h) \setminus \{r\}}) = \argmin_{a_r} Y_h({\mv{a}_{dep(\mv{C}_h)\setminus \{ r \}},a_r})$,
\item ${Y}^{\mt{new}}_h(\mv{a}_{dep^{\mt{new}}(\mv{C}_h)}) = Y^{\mt{new}}_h(\mv{a}_{dep(\mv{C}_h) \setminus \{r\}}) = Y_h(\mv{a}_{dep(\mv{C}_h)\setminus \{r\}},z_h(\mv{a}_{dep(\mv{C}_h) \setminus \{r\}}))$.
\end{enumerate}

Based on the above definition of $Y_h^{\mt{new}}$, we know that $f^{\mt{new}}_r(\mv{X}^{\prime}) \leq f_{r,p}(\mv{X}^{\prime})$.  In the remainder of the proof, we show that $\mv{X}^{\prime}$ is also a feasible solution for $P^{\mt{new}}_{r,p}$.  The first set of constraints of $P^{\mt{new}}_{r,p}$ are valid (as we have not modified $\mv{X}$).  Similarly, the second set of constraints is valid for all $j \neq h$ (as we have not changed any parameters).  Now we prove the constraints in \eqref{constraint2}, for $j=h$.  For all assignments $a_{ dep^{\mt{new}}(\mv{C}_h) \setminus \mv{C}_{h}}$,
\begin{align*}
\sum_{a_{n_{h,1}}}& X^{\prime}_{n_{h,1}} (\mv{a}_{dep(n_{h,1})}) \cdots \sum_{a_{n_{h,s_h}}} X^{\prime}_{n_{h,s_h}} (\mv{a}_{dep(n_{h,s_h})})  
\cdot
{Y}^{\mt{new}}_h (\mv{a}_{dep^{\mt{new}}({C}_h)})  \\
	&= \sum_{a_{n_{h,1}}} X^{\prime}_{n_{h,1}} (\mv{a}_{dep(n_{h,1})}) \cdots \sum_{a_{h_{s_j}}} X^{\prime}_{n_{h,s_h}} (\mv{a}_{dep(n_{h,s_h})})  
\cdot
Y_h (\mv{a}_{dep(\mv{C}_h) \setminus \{ r \}},a_r = z_h(\mv{a}_{dep(\mv{C}_h)\setminus \{r\}})) \\ 
& \hspace{25em} \text{(by definition of $Y_h^{\mt{new}}$)} \\
&= \left[ \sum_{a_{n_{h,1}}} X^{\prime}_{n_{h,1}} (\mv{a}_{dep(n_{h,1})}) \cdots \sum_{a_{h_{s_j}}} X^{\prime}_{n_{h,s_h}} (\mv{a}_{dep(n_{h,s_h})})  
\cdot
\sum_{a_r} Y_h (\mv{a}_{dep(\mv{C}_h \setminus \{ r \})},a_r) \right] \\ 
& \quad \, - \left[ \sum_{a_{n_{h,1}}} X^{\prime}_{n_{h,1}} (\mv{a}_{dep(n_{h,1})}) \cdots \sum_{a_{h_{s_j}}} X^{\prime}_{n_{h,s_h}} (\mv{a}_{dep(n_{h,s_h})})  
\cdot
\sum_{a_r:a_r \neq z_h(\mv{a}_{dep(\mv{C}_h) \setminus \{ r \}})} Y_h (\mv{a}_{dep(\mv{C}_h) \setminus \{ r \}},a_r) \right] \\
&\geq \left[ \sum_{a_r} \sum_{a_{n_{h,1}}} X^{\prime}_{n_{h,1}} (\mv{a}_{dep(n_{h,1})}) \cdots \sum_{a_{n_{h,s_h}}} X^{\prime}_{n_{h,s_h}} (\mv{a}_{dep(n_{h,s_h})})  
\cdot
Y_h (\mv{a}_{dep(\mv{C}_h) \setminus \{ r \}},a_r) \right] \\
&\quad \, - \left[ \sum_{a_{n_{h,1}}} X^{\prime}_{n_{h,1}} (\mv{a}_{dep(n_{h,1})}) \cdots \sum_{a_{h_{s_j}}} X^{\prime}_{n_{h,s_h}} (\mv{a}_{dep(n_{h,s_h})})  
\cdot
\sum_{a_r:a_r \neq z(\mv{a}_{dep(\mv{C}_h) \setminus \{ r \}})} 1 \quad \right] (\because Y_h(.) \leq 1) \\
	&\geq \left[ \sum_{a_r \in \Sigma} (1-b_h \gamma) \right] -  \left[(|\Sigma - 1|) \cdot \sum_{a_{n_{h,1}}} X^{\prime}_{n_{h,1}} (\mv{a}_{dep(n_{h,1})}) \cdots \sum_{a_{h_{s_j}}} X^{\prime}_{n_{h,s_h}} (\mv{a}_{dep(n_{h,s_h})}) \right] \\
 & \hspace{25em} \text{(by constraint \eqref{constraint2} of $P_{r,p}$)} \\
&\geq |\Sigma| (1 - b_h  \gamma) - (|\Sigma|-1) 1 \quad \text{(by constraint \eqref{constraint1} of $P_{r,p}$)} \\
&= 1 - |\Sigma|  b_h \gamma \\
&=1 - b_h^{\mt{new}} \gamma.
\end{align*}
\end{proof}

After multiple passes through the above lemma, we get a program $P_{r,p}$ that satisfies $r \notin dep(\mv{C}_j)$, for all $\mv{C}_j$ such that $r \notin \mv{C_j}$.  The next lemma takes in such a program, and outputs a program with a reduced value of $r$.

\begin{lemma}[R-Elimination] \label{lemma:R-Elimination}
Let $P_{r,p}(\Sigma,\gamma,\mcl{C},\mathbf{b},\mathbf{dep},\mathbf{Y})$ be a given program such that the element $r \in \mv{C}_f$.  Suppose $dep(r)=dep(\mv{C}_f)$, and for all $j \in [p]\setminus \{f\}$, $r \notin dep(\mv{C}_j)$.  Then there exists a program $P^{\mt{new}}_{r-1,p}(\Sigma,\gamma,\mathcal{C}^{\mt{new}},\mathbf{b},\mathbf{dep}^{\mt{new}},\mathbf{Y}^{\mt{new}})$ such that 
\begin{align*}
\Opt(P_{r,p}) \geq \Opt(P^{\mt{new}}_{r-1,p}) 
\end{align*}
where $\mathbf{Y}^{\mt{new}}$ differs from $\mathbf{Y}$ only on the function $Y_f$, $\mathcal{C}^{\mt{new}}$ differs from $\mathcal{C}$ only on the partition $\mv{C}_{f}$ where $\mv{C}_{f}^{\mt{new}} = \mv{C}_{f} \setminus \{ r \}$, and $\mathbf{dep}^{\mt{new}} = (dep(1),dep(2),\ldots,dep(r-1),dep(\mv{C}_1),\ldots, dep(\mv{C}_{f -1}),$ $dep^{\mt{new}}(\mv{C}^{\mt{new}}_f),$ $dep(\mv{C}_{f+1}),$ $\dots,$ $dep(\mv{C}_p))$ where $dep^{\mt{new}}(\mv{C}^{\mt{new}}_f) = dep(r)\setminus \{r\}$. 
\end{lemma}

\begin{proof}
Let $\mv{X}^{\mt{old}}$ be an optimal solution of $P_{r,p}$.  For a given set of functions $\mv{X}$, a subset $\mv{S} \subseteq [r]$,  and for a given assignment $a_{\mv{S}}$ to $\mv{S}$, let $f_{r,p}( \mv{X} ) |_{a_{\mv{S}}}$ represent the sum of all terms in $f_{r,p}(\mv{X})$ that are consistent with the assignment $a_{\mv{S}}$.  Then, for all assignments $\mv{a}_{dep(r) \setminus \{ r \}}$
\begin{align}
f_{r,p}(\mv{X}^{\mt{old}})|_{\mv{a}_{dep(r) \setminus \{ r \}}} = L_{\mv{a}_{dep(r) \setminus \{ r \}}} \cdot \sum_{a_r} X^{\mt{old}}_r (\mv{a}_{dep(r)})\cdot Y_f({\mv{a}_{dep(r)}}) \label{eqn:restrictedObjective}
\end{align}
We define $Y_{f}^{\mt{new}}(\mv{a}_{dep^{\mt{new}}(\mv{C}^{\mt{new}}_f)}) = Y^{\mt{new}}_{f}(\mv{a}_{dep(r)\setminus \{ r \}}) = \sum_{a_r} X^{\mt{old}}_r (\mv{a}_{dep(r)}) \cdot Y_f({\mv{a}_{dep(r)}})$. Observe that $Y^{\mt{new}}_f:\Sigma^{|dep(r)\setminus \{r\}|} \to [0,1]$ because of constraint \eqref{constraint1} and since $Y_f$ itself falls in the range $[0,1]$.
Now, the new program $P^{\mt{new}}_{r-1,p}$ is completely specified. 

Observe that:
\ignore{ Note that $L_{\mv{a}_{dep(r)\setminus \{r\}}}$ does not depend on the value of ${a}_r$.  Therefore, when $\mv{a}_{dep(r) \setminus \{ r \} }$ is fixed, with respect to each value of ${a}_r \in \Sigma$, \eqref{eqn:restrictedObjective} consists of the terms $L_{\mv{a}_{dep(r) \setminus \{ r \}}} \cdot X^{\mt{old}}_r (\mv{a}_{dep(r)}) \cdot Y_f({\mv{a}_{dep(r)}})$.  Hence, }
\begin{align*}
f_{r,p}(\mv{X}^{\mt{old}})|_{\mv{a}_{dep(r) \setminus \{ r \}}} 
			= L_{\mv{a}_{dep(r) \setminus \{ r \}}} \cdot Y^{\mt{new}}_{f}(\mv{a}_{dep(r)\setminus \{ r \}})
			&= f_{r-1,p}^{\mt{new}}(\mv{X}_{r-1}^{\mt{old}})|_{\mv{a}_{dep(r) \setminus \{ r \}}}
\end{align*}
where $\mv{X}^{\mt{old}}_{r-1} = \{ X^{\mt{old}}_1,\ldots,X^{\mt{old}}_{r-1} \}$.  This implies $$f_{r,p}(\mv{X}^{\mt{old}}) = \sum_{\mv{a}_{dep(r)\setminus \{ r \}}} f_{r,p}(\mv{X}^{\mt{old}})|_{\mv{a}_{dep(r) \setminus \{ r \}}} =\sum_{\mv{a}_{dep(r)\setminus \{ r \}}} f_{r-1,p}^{\mt{new}}(\mv{X}_{r-1}^{\mt{old}})|_{\mv{a}_{dep(r) \setminus \{ r \}}} = f_{r-1,p}^{\mt{new}}(\mv{X}^{\mt{old}}_{r-1}).$$

We now show that the functions $\mv{X}^{\mt{old}}_{r-1}$ form a feasible solution for $P^{\mt{new}}_{r-1,p}$. 
The first set of constraints \eqref{constraint1} holds for $P^{\mt{new}}_{r-1,p}$ because $\mv{X}^{\mt{old}}$ is feasible for $P_{r,p}$.  Also for all $j \neq f$, the second set of constraints \eqref{constraint2} holds for the same reason.  For $j=f$: 
\begin{align*}
		&\sum\limits_{a_{n_{f, 1}} \in \Sigma} X^{\mt{old}}_{n_{f,1}}(\mv{a}_{dep(n_{f,1})}) \cdots \sum\limits_{a_{n_{f, s_f-1}} \in \Sigma} X^{\mt{old}}_{n_{f, s_f-1}}(\mv{a}_{dep(n_{f,s_f-1})}) \cdot Y_f^{\mt{new}}({\mv{a}_{dep(C^{\mt{new}}_f)}}) \\
	&= 		\sum\limits_{a_{n_{f, 1}} \in \Sigma} X^{\mt{old}}_{n_{f,1}}(\mv{a}_{dep(n_{f,1})}) \cdots \sum\limits_{a_{n_{f, s_f-1}} \in \Sigma} X^{\mt{old}}_{n_{f, s_f-1}}(\mv{a}_{dep(n_{f,s_f-1})})  \sum_{a_r \in \Sigma} X^{\mt{old}}_r(\mv{a}_{dep(r)}) \cdot Y_f(\mv{a}_{dep(\mv{C}_f)}) \\ & \hspace{25em} \text{(by definition)} \\
	&\geq 1-b_f \gamma
\end{align*}
This completes the proof that $\Opt(P_{r,p}) \geq \Opt(P^{\mt{new}}_{r-1,p})$.
\end{proof}

\begin{lemma} \label{programPHelperLemma}
For any integers $r,p \geq 1$ and given a program $P_{r,p}(\Sigma,\gamma,\mcl{C},\mathbf{b},\mathbf{dep},\mathbf{Y})$, there exists a 
program $P^{\mt{new}}_{r-1,p}(\Sigma,$ $\gamma,$ ${\mcl{C}}^{\mt{new}},$ ${\mathbf{b}^{\mt{new}}},$ $\mathbf{dep}^{\mt{new}},$ ${\mathbf{Y}}^{\mt{new}})$ such that 
\begin{align*}
\Opt(P_{r,p}) \geq \Opt(P^{\mt{new}}_{r-1,p})
\end{align*}
where
\begin{align*}
&{b}^{\mt{new}}_j  = b_j \quad 
 	& &,\, \forall j \in [p] \colon r \notin dep(\mv{C}_j), \\
&{b}^{\mt{new}}_j  = |\Sigma| \cdot b_j \quad 
	& &, \, \forall j \in [p] \colon r \in dep(\mv{C}_j).
\end{align*}
\end{lemma}

\begin{proof}
First we apply Lemma~\ref{lemma:dependent-set-reduction} (Dependent Set Reduction).  Then, we apply Lemma~\ref{lemma:Y-Rreduction} (Y-R Reduction) repeatedly, until there does not exist any $h \in [p]$ such that $r \in dep(\mv{C}_h)$ but $r \notin \mv{C}_h$.  Note that, in each step of this reduction, the respective $b_h$ increases by a factor of $|\Sigma|$.  Finally, applying Lemma \ref{lemma:R-Elimination} (R-Elimination) results in a program $P^{\mt{new}}_{r-1,p}$  on $r-1$ inputs with the desired property.
\end{proof}

\begin{lemma} \label{programPMainLemma}
For a given program $P_{r,p}(\Sigma,\gamma,\mcl{C},\mathbf{b} = \mathbf{1},\mathbf{dep},\mathbf{Y} = \mathbf{1})$, suppose we know that $\max\limits_{j} |dep(\mv{C}_j)|$ is at most $L$.  Then
\begin{align*}
\Opt(P_{r,p}) \geq (1 - {|\Sigma|}^{L} \gamma)^{p}.
\end{align*}
\end{lemma}

\begin{proof}
We apply Lemma \ref{programPHelperLemma} recursively.  Note that in each such reduction from $P_{r,p}$ to $P_{r-1,p}$, the value of $b_j$ increases by a factor of $|\Sigma|$ only when $r \in dep(\mv{C}_j)$.  

At $r=0$, we have the program $P_{0,p}(\Sigma, \gamma, \mcl{C}', \mv{b}', \mv{dep}', \mv{Y}')$. 
For all $j \in [p]$, we know that $b'_j \leq |\Sigma|^{L}$ (since $|dep(\mv{C}_j)| \leq L$).  Therefore,
\begin{align*}
\Opt(P_{r,p}) & \geq \Opt(P_{0,p}) \\
		 & = \prod_{j=1}^{p} Y'_j(\emptyset)\\
		 & \geq \prod_{j=1}^{p} (1 - b'_j  \gamma) \quad \text{based on constraint \eqref{constraint2} of the program $P_{0,p}$} \\
		 & \geq (1 - |\Sigma|^{L} \gamma)^{p}.
\end{align*}
\end{proof}

\section{Acknowledgments}
We would like to thank Vasant Honavar who told us about the problems considered here and for several helpful discussions that were essential for us to complete this work.

\ignore{ 
\section{Testing Bayesian Networks with Hidden Variable}

\begin{theorem} \label{theorem:hiddenVariables}
Let $P$ and $Q$ be two SMCMs defined on a {\em known and common} SMCG $G$ of maximum degree $d$ and maximum c-component size $\ell$.  Let $\{V_1,\ldots,V_n\}$ be a topological ordering and let $\mv{C}_{V_i}$ be the c-component of $G$ that contains $V_i$, and $\mv{V}^{(i)}$ be the set $\{V_1,\ldots,V_i\}$.  Then
\begin{align*}
H^2\left(P[V_1,\ldots,V_n],Q[V_1,\ldots,V_n]\right) \leq \sum_{i=1}^{n} H^2\left(P[V_i,\Neighbor(V_i)],Q[V_i,\Neighbor(V_i)]\right) 
\end{align*}
where we call $\Neighbor(V_i)= (\mv{C}_{V_{i}} \cap \mv{V}^{(i-1)})\cup \Pa(\mv{C}_{V_{i}} \cap \mv{V}^{(i)})$, the neighbors of $V_i$.  
\end{theorem}

\begin{proof}
\begin{align*}
H^2(P,Q) &= 1 - \sum_{\mv{v}^{(n)}} \sqrt{ P[\mv{v}] Q[\mv{v}] } \\
&= 1 - \sum_{\mv{v}^{(n-1)}} \sqrt{ P[\mv{v}^{(n-1)}] Q[\mv{v}^{(n-1)}]} \sum_{v_{n}} \sqrt{ P[v_n \mid \mv{v}^{(n-1)}] Q[v_n \mid \mv{v}^{(n-1)}]}
\end{align*}
\end{proof}
By Lemma~\ref{lemma:conditional_independence}, we know that $V_n$ is independent of $\mv{V}^{(n-1)} \setminus \Neighbor(V_n)$ when conditioned on $\Neighbor(V_n)$.  Therefore,
\begin{proof}
\begin{align*}
H^2(P,Q) &= 1 - \sum_{\mv{v}^{(n-1)}} \sqrt{ P[\mv{v}^{(n-1)}] Q[\mv{v}^{(n-1)}] } \sum_{v_{n}} \sqrt{ P[v_n \mid \neighbor(v_n)] Q[v_n \mid \neighbor(v_n)]} \\
		 &= 1 - \begin{array}{l} 
		 	\sum\limits_{\mv{v}^{(n-1)}} \left( \dfrac{P[\mv{v}^{(n-1)}] + Q[\mv{v}^{(n-1)}]}{2} \right) \sum\limits_{v_{n}} \sqrt{ P[v_n \mid \neighbor(v_n)] Q[v_n \mid \neighbor(v_n)}] \\ \newline 
		 	+ \sum\limits_{\mv{v}^{(n-1)}} \left( \dfrac{P[\mv{v}^{(n-1)}] + Q[\mv{v}^{(n-1)}]}{2} - \sqrt{ P[\mv{v}^{(n-1)}] Q[\mv{v}^{(n-1)}] } \right) \sum\limits_{v_{n}} \sqrt{ P[v_n \mid \neighbor(v_n)] Q[v_n \mid \neighbor(v_n)]} 
		 	\end{array} \\		 	
		 &= 1 - \begin{array}{l} 
		 	\sum\limits_{\neighbor(V_n)} \left( \dfrac{P[\neighbor(V_n)] + Q[\neighbor(V_n)]}{2} \right) \sum\limits_{v_{n}} \sqrt{ P[v_n \mid \neighbor(V_n)] Q[v_n \mid \neighbor(V_n)]} \\ \newline
		 	+ \sum\limits_{\mv{v}^{(n-1)}} \left( \dfrac{P[\mv{v}^{(n-1)}] + Q[\mv{v}^{(n-1)}]}{2} - \sqrt{ P[\mv{v}^{(n-1)}] Q[\mv{v}^{(n-1)}] } \right) \sum\limits_{v_{n}} \sqrt{ P[v_n \mid \neighbor(V_n)] Q[v_n \mid \neighbor(V_n)]} 
		 	\end{array} \\ 
		 &\leq 1 - \begin{array}{l} 
		 	\sum\limits_{\neighbor(V_n)} \left( \sqrt{ P[\neighbor(V_n)] Q[\neighbor(V_n)] } \right) \sum\limits_{v_{n}} \sqrt{ P[v_n \mid \neighbor(V_n)] Q[v_n \mid \neighbor(V_n)]}  \\ \newline
		 	+ \sum\limits_{\mv{v}^{(n-1)}} \left( P[\mv{v}^{(n-1)}] + Q[\mv{v}^{(n-1)}] - \sqrt{ P[\mv{v}^{(n-1)}] Q[\mv{v}^{(n-1)}] } \right) 
		 	\end{array} \quad \text{(using Cauchy Schartz)} \\
		 &\leq 1 - \sum_{\neighbor(V_n),v_n} \sqrt{ P[v_n,\neighbor(V_n)] Q[v_n,\neighbor(v_n)]} + \sum_{\mv{v}^{(n-1)}} \left( \sqrt{P[\mv{v}^{(n-1)}]} - \sqrt{Q[\mv{v}^{(n-1)}]}\right)^2 \\
		 &= H^2(P[V_n,\Neighbor(V_n)],Q[V_n,\Neighbor(V_n)]) + H^2(P[\mv{V}^{(n-1)}],Q[\mv{V}^{(n-1)}]) 
\end{align*} 		 
By applying the above inequality recursively we get the required result.
\end{proof}

\begin{lemma} \label{lemma:conditional_independence}
For a given SMCM $P$ defined on a SMCG $G$, let $\mv{V} = \mv{V}^{(n)}$ be a topological ordering of $G$.  Let $\mv{C}_{V_j}$ be the c-component of $G$ that contains $V_j$.  Then for a given $V_i \in \mv{V}$, and a given assignment $\mv{v}^{(i)}$ to $\mv{V}^{(i)}$,
$$ P[v_i  \mid  \mv{v}^{(i-1)}] = P[ v_i  \mid  \neighbor(V_i) ]$$
where $\neighbor(V_i) = (\mv{C}_{V_i} \cap \mv{V}^{(i-1)}) \cup \Pa((\mv{C}_{V_i} \cap \mv{V}^{(i)}))$ is the assignment to $\Neighbor(V_i)$ that is consistent with $\mv{v}^{(i-1)}$.
\end{lemma}

\begin{proof}
Using Bayes' rule, 
$$ P[v_i  \mid  \mv{v}^{(i-1)}] = \dfrac{P[\mv{v}^{(i)}]}{P[\mv{v}^{(i-1)}]} $$
Because of the topological ordering of $V$, intervening the nodes in $\mv{V}^{(i+1,n)} \defeq \{ V_{i+1},\ldots,V_n \}$ does not affect the distribution of $\mv{V}^{(i)}$.  Therefore, for an arbitrary assignment $\mv{v}^{(i+1,n)}$ to $\mv{V}^{(i+1,n)}$,
$$ P[v_i  \mid  \mv{v}^{(i-1)}] = \dfrac{P[\mv{v}^{(i)} \mid do(\mv{v}^{(i+1,n)})]}{P[\mv{v}^{(i-1)} \mid do(\mv{v}^{(i+1,n)})]} = P[V_i \mid \mv{V}^{(i-1)},do(\mv{v}^{(i+1,n)})]$$
Let $G^{\prime} = G_{\mv{\overline{V}}^{(i+1,n)}}$ be the induced graph with respect to the intervention $do(\mv{v}^{(i+1,n)})$ obtained from $G$ by removing the incoming edges of $\mv{V}^{(i+1,n)}$.  Let $\mv{W} = \mv{V}^{(i-1)} \setminus \Neighbor(V_i)$.  First we prove that $[V_i \indep \mv{W}  \mid  \Neighbor(V_i)]$ in the graph $G^{\prime}$.  Suppose, for contradiction, there exists a path in $G^{\prime}$ between some nodes of $\mv{W}$ and $V_i$ that is not blocked.  First note that $p$ cannot contain nodes from $\mv{V}^{(i+1,n)}$ (since incoming edges to $\mv{V}^{(i+1,n)}$ are removed in $G^{\prime}$).  Let $W_j \in \mv{W}$ be the $W_j$ that is closest to $V_i$ in $p$.  Now $W_j \rightarrow V_i$ or $W_j \leftrightarrow V_i$ implies $W_j \in \Neighbor(V_i)$, which is not possible.  Also, $W_j \leftarrow V_i$ implies $W_j \in \mv{V}^{(i+1,n)}$, which is not possible.  Therefore, the nodes in between $W_j$ and $V_i$ in $p$ belong to $\Neighbor(V_i)$ and such nodes must be collider nodes (as otherwise $p$ would be blocked).  The only two possibilities are $W_j \rightarrow () \leftrightarrow () \leftrightarrow \cdots \leftrightarrow V_i$ or $W_j \leftrightarrow () \leftrightarrow () \leftrightarrow \cdots \leftrightarrow V_i$, either of them implies $W_j \in \Neighbor(V_i)$, a contradiction.  Therefore, $[V_i \indep \mv{W}  \mid  \Neighbor(V_i)]$ in the graph $G^{\prime}$.  Using this conditional independence we get,        
$$ P[v_i  \mid  \mv{v}^{(i-1)}] = P[v_i \mid \neighbor(V_i),do(\mv{v}^{(i+1,n)})]$$
Since intervening vertices in $\mv{V}^{(i+1,n)}$ does not affect the distribution on $V_i,\Neighbor(V_i)$, we get 
$$ P[v_i  \mid  \mv{v}^{(i-1)}] = P[v_i \mid \neighbor(V_i)].$$
\end{proof}  
} 

\bibliographystyle{alpha}
\bibliography{confs}

\appendix

\section{Proof Sketch for the Fully Observable Case}\label{sec:fully_observable_case}
In the absence of unobservable variables, the analysis becomes much simpler. Let us look at the two-sample testing problem on input causal models $\mathcal{X}$ and $\mathcal{Y}$ defined on a DAG $G$. Now, each c-component is a single vertex, so that every ``local'' intervention is of the form $P[V_i \mid do(\pa(V_i))]$ for a vertex $V_i$ and an assignment $\pa(V_i)$ to the parents of $V_i$. We define our tester to accept iff each such local intervention on $\mathcal{X}$ and $\mathcal{Y}$ yields distributions which differ by at most $\eps^2/2n$ in squared Hellinger distance.  The squared Hellinger distance is defined as follows for two  distributions $P$ and $Q$ on $[D]$:
\begin{align}
H^2(P,Q) :&= 1-\sum_{i\in [D]}\sqrt{P(i)\cdot Q(i)}  = 1 - BC(P,Q)\label{eqn:squaredHellinger}
\end{align}
where $BC(P,Q)$ is the Fidelity or Bhattacharya coefficient of $P$ and $Q$.  Below, our subadditivity theorem shows that if the algorithm accepts, then for every intervention, the resulting distributions for $\mathcal{X}$ and $\mathcal{Y}$ differ by at most  $\eps^2/2$ in squared Hellinger distance, implying $\Delta(\mathcal{X},\mathcal{Y}) \leq \eps$.

\begin{theorem} \label{theorem:mainTheorem_without_hidden_vars}
Let $\mathcal{X}$ and $\mathcal{Y}$ be two causal Bayesian networks defined on a {\em known and common} DAG $G$ with no hidden variables.  Identify the vertices in $\mathbf{V}$ as $\{V_1, \dots, V_n\}$ arranged in a topological order.  Suppose we know that 
\begin{align}
H^2( P_{\mathcal{X}}[V_j \mid  do(\pa(V_j))] , P_{\mathcal{Y}}[V_j  \mid  do(\pa(V_j))] ) \leq \gamma \qquad \forall j \in [n],\forall \pa(V_j) \in \Sigma^{|\Pa(V_j)|}. \label{known_1}
\end{align}
Then, for each subset $\mathbf{T} \subseteq \mathbf{V}$ and $\mathbf{t} \in \Sigma^{|\mathbf{T}|}$,
\begin{align}
H^2 \left( P_{\mathcal{X}}[\mathbf{V \setminus T}  \mid  do(\mathbf{t}) ] , P_{\mathcal{Y}}[\mathbf{V \setminus T}  \mid  do(\mathbf{t}) ] \right) \leq {\gamma} n. \label{unknown_1}
\end{align}
\end{theorem}

\begin{proof}
Fix $\mathbf{T \subseteq V}$ and an assignment $\mathbf{t} \in \Sigma^{|\mathbf{T}|}$.  Let $\mathbf{W} = \mathbf{V \backslash T}=\{W_1,W_2,\ldots,W_m\}$ whose indices are arranged in a topological ordering.  By the definition of squared Hellinger distance:
\begin{align*}
H^2 &\left(\begin{array}{l} P_{\mathcal{X}}[\mathbf{W} |do(\mathbf{t})],\\ \newline P_{\mathcal{Y}}[\mathbf{W} |do(\mathbf{t})] \end{array}\right) \\
	&= 1 - \sum\limits_{w_1,w_2,\ldots,w_m} \sqrt{\begin{array}{l}P_{\mathcal{X}} [w_1,w_2,\ldots,w_m | do(\mathbf{t})] \\ \newline P_{\mathcal{Y}} [w_1,w_2,\ldots,y_m | do(\mathbf{t})] \end{array}} \\
	&= 1 - \sum\limits_{w_1,\ldots,w_{m-1}} \sqrt{\begin{array}{l}P_{\mathcal{X}} [w_1,\ldots,w_{m-1} | do(\mathbf{t})] \\ \newline P_{\mathcal{Y}} [w_1,\ldots,w_{m-1} | do(\mathbf{t})] \end{array}} \sum\limits_{w_m} \sqrt{\begin{array}{l} P_{\mathcal{X}}[w_m|w_1,\ldots,w_{m-1},do(\mathbf{t})] \\ \newline P_{\mathcal{Y}}[w_m|w_1,\ldots,w_{m-1},do(\mathbf{t})] \end{array}}  \\
	&= 1 - \sum\limits_{w_1,\ldots,w_{m-1}} \sqrt{ \begin{array}{l}P_{\mathcal{X}} [w_1,\ldots,w_{m-1} | do(\mathbf{t})] \\ \newline P_{\mathcal{Y}} [w_1,\ldots,w_{m-1} | do(\mathbf{t})] \end{array}} \sum\limits_{w_m} \sqrt{\begin{array}{l} P_{\mathcal{X}}[w_m|do(\pa(w_m))] \\ \newline P_{\mathcal{Y}}[w_m|do(\pa(w_m))] \end{array}}.	
\end{align*}
The above step can be obtained easily by using Lemma \ref{lemma:independenceLemma} and the conditional independence constraints obtained from $G$.  Therefore:
\begin{align*}
H^2 \left(\begin{array}{l} P_{\mathcal{X}}[\mathbf{W} |do(\mathbf{t})],\\ \newline P_{\mathcal{Y}}[\mathbf{W} |do(\mathbf{t})] \end{array}\right)
	&\leq 1 - \sum\limits_{w_1,\ldots,w_{m-1}} \sqrt{ \begin{array}{l}P_{\mathcal{X}} [w_1,\ldots,w_{m-1} | do(\mathbf{t})] \\ \newline P_{\mathcal{Y}} [w_1,\ldots,w_{m-1} | do(\mathbf{t})] \end{array}} (1 - \gamma) \quad (\text{from}~\eqref{known_1})\\
	&=H^2 \left(\begin{array}{l}  P_{\mathcal{X}}[W_1 \ldots W_{m-1} \mid do(\mathbf{t})],\\ \newline P_{\mathcal{Y}}[W_1 \ldots W_{m-1} |do(\mathbf{t})] \end{array} \right) (1-\gamma) + \gamma. 
\end{align*}
By induction on $n$, we get:
\begin{align*}
H^2 \left(\begin{array}{l} P_{\mathcal{X}}[\mathbf{W} |do(\mathbf{t})],\\ \newline P_{\mathcal{Y}}[\mathbf{W} |do(\mathbf{t})] \end{array}\right)  &\leq \gamma [1 + (1-\gamma) + (1-\gamma)^2 +\ldots+ (1-\gamma)^{m-1}] \\
&= 1 - (1-\gamma)^{m} \leq 1 - (1-\gamma)^{n} \leq n \gamma.
\end{align*}     
\end{proof}
The time and sample complexities are then determined by that required for two-sample testing on each pair of  local distributions with accuracy $\eps^2/2n$ in $H^2$ distance. We defer this calculation, as well as bounding the total number of interventions, to later when we analyze semi-Markovian CBNs.

\section{Reduction from General Graphs} \label{section:general-graph-reduction}

First we define the effective parents and the c-component relation for general causal graphs.

\begin{definition}[Effective Parents $\ePa$]
Given a general causal graph $H$ and a vertex $V_i \in \mv{V}$, {\em the effective parents of $V_i$}, denoted by $\ePa(V_i)$, is the set of all observable vertices $V_j$ such that either $V_j$ is a parent of $V_i$ or there exists a directed path from $V_j$ to $V_i$ that contains only unobservable variables.
\end{definition}

\begin{definition}[c-component]
For a given general causal graph $H$, two vertices $V_i$ and $V_j$ are related by the {\em c-component relation} if (i) there exists an unobservable variable $U_k$ such that $H$ contains two paths (i) from $U_k$ to $V_i$; and (ii) from $U_k$ to $V_i$, where both the paths use only unobservable variables, or (ii) there exists another vertex $V_{z} \in \mv{V}$ such that $V_i$ and $V_{z}$ (and) $V_j$ and $V_{z}$ are related by c-component relation.
\end{definition}

We study Semi Markovian Bayesian Networks ({\smbn})'s without any loss of generality owing to the projection of a general causal graph to a \smcg~\cite{tian-pearl,verma-pearl}. For a given graph $H$ they showed that there is an equivalent {\smcg} $G$ such that the c-component factorization and some other important properties hold.  Namely,
\begin{itemize}
\item 
The set of observable nodes in $H$ and $G$ are the same.
\item
The topological ordering of the observable nodes in $H$ and $G$ are the same.
\item
The $c$-components of $H$ and $G$ are identical and the c-component factorization formula (Lemma \ref{lemma:c-component-factorization} here, ($20$) in Lemma~$2$ of \cite{tian-pearl}) holds even for the general causal graph (See Section~5 of \cite{tian-pearl}).  They show this based on a known previously known reduction from $H$ to $G$ \cite{verma-pearl}.  The proof is based on the fact that for any subset $\mv{S} \subseteq \mv{V}$ of observable variables, the induced subgraphs $G[\mv{S}]$ and $H[\mv{S}]$ require the same set of conditional independence constraints.
\item
The parents of nodes in $G$ are the effective parents of nodes in $H$. 
\end{itemize}

All the results presented in this paper depend only on the above mentioned properties.  Therefore, we can reduce the given general causal graph $H$ to a {\smcg} $G$ using the available reduction and work with $G$, where the parents of vertices of $G$ correspond to the effective parents of the respectives vertices of $H$.  Now we proceed to show the algorithm of \cite{verma-pearl} that preserves all the required properties mentioned above.  
\noindent \paragraph{Projection Algorithm of \cite{tian-pearl,verma-pearl}} For a given causal graph $H$, the projection algorithm reduces the given causal graph $H$ to a {\smcg} $G$ by the following procedure:
\begin{enumerate}
\item For each observable variable $V_i \in V$ of $H$, add an observable variable $V_i$ in $G$.
\item For each pair of observable variables $V_i,V_j \in \mv{V}$, if there exists a directed edge from $V_i$ to $V_j$ in H, or if there exists a {\em directed} path from $V_i$ to $V_j$ that contains only unobservable variables in $H$, then add a directed edge from $V_i$ to $V_j$ in $G$.  
\item For each pair of observable variables $V_i,V_j \in \mv{V}$, if there exists an unobservable variable $U_k$ such that there exist two {\em directed} paths in $H$ from $U_k$ to $V_i$ and from $U_k$ to $V_j$ such that both the paths contain only the unobservable variables, then add a bi-directed edge between $V_i$ and $V_j$ in $G$.
\end{enumerate}

\section{Conditional Independence}

The following lemma captures a useful fact about conditional independence between variables in a \smbn.
\begin{lemma}[Independence Lemma] \label{lemma:independenceLemma}
Let $M$ be a {\smbn} with respect to a {\smcg} $G$ with the vertex set $\mv{V}=\{ V_1, \ldots, V_n \}$ (where the indices respect topological ordering).  For a given intervention $do(\mv{t})$, let $\mv{C} = \{ V_{n_{1}}, V_{n_{2}},\ldots, V_{n_{s}} \}$ be a c-component of the induced subgraph $G^{\prime} = G[\mv{V} \setminus \mv{T}]$, where $s = |\mv{C}|$ and $n_{1} < n_{2} < \cdots < n_{s}$.  Then for a given vertex $V_{n_{i}}$, for a given set $\mv{D}$ such that $\mv{V} \setminus (\mv{T} \cup \{V_{n_{1}},\ldots, V_{n_{i}}\}) \supseteq \mv{D} \supseteq \Pa_{G^{\prime}}(\{ V_{n_{1}}, \ldots, V_{n_{i}} \})$, and a given set of assignments $v_{n_{1}}, \ldots ,v_{n_{i}}$, $\mv{d}$,
$$P_{\mathcal{M}}[v_{n_{i}}  \mid  v_{n_{1}}, \ldots ,v_{n_{i-1}} , do(\mv{d},\mv{t})] = P_{\mathcal{M}}[ v_{n_{i}}  \mid  v_{n_{1}}, \ldots, v_{n_{i-1}}, do(\pa\nolimits_{G^{\prime}}(V_{n_{1}},\ldots, V_{n_{i}}),\mv{t} )  ]$$
where $pa_{G^{\prime}}(v_{n_{1}},\ldots, v_{n_{i}})$ is the assignment that is consistent with $\mv{D}$.
\end{lemma}

\begin{proof}
By Bayes' theorem 
\begin{align} \label{eqn:independence}
P_{\mathcal{M}}\left[v_{n_{j,i}} \middle| \begin{array}{l} v_{n_{j,1}},\ldots,v_{n_{j,i-1}}, \\ \newline do(\pa\nolimits_{G^{\prime}}(V_{n_{j,1}},\ldots, V_{n_{j,i}}), \mv{t}) \end{array}\right] = \dfrac{P_{\mathcal{M}}[v_{n_{j,i}}, v_{n_{j,1}},\ldots,v_{n_{j,i-1}} \mid do(\pa\nolimits_{G^{\prime}}(V_{n_{j,1}},\ldots, V_{n_{j,i}}),\mv{t}) ]} {P_{\mathcal{M}}[v_{n_{j,1}},\ldots,v_{n_{j,i-1}} \mid do(\pa\nolimits_{G^{\prime}}(V_{n_{j,1}},\ldots, V_{n_{j,i}}),\mv{t}) ]}.
\end{align}
We apply Lemma~\ref{lemma:trimUnnecessaryActions} with respect to the graph $G^{\prime} = G[\mv{V} \setminus \mv{T}]$ that is obtained after the intervention $do(\mv{t})$ for both the numerator and the denominator of \eqref{eqn:independence} seperately.  Therefore:
\begin{align*}
P_{\mathcal{M}}\left[v_{n_{j,i}} \middle| \begin{array}{l} v_{n_{j,1}},\ldots,v_{n_{j,i-1}}, \\ \newline do(\pa\nolimits_{G^{\prime}}(V_{n_{j,1}},\ldots, V_{n_{j,i}}),\mv{t}) \end{array} \right] &= \dfrac{P_{\mathcal{M}}[v_{n_{j,i}}, v_{n_{j,1}},\ldots,v_{n_{j,i-1}} \mid do(\mv{d},\mv{t}) ]} {P_{\mathcal{M}}[v_{n_{j,1}},\ldots,v_{n_{j,i-1}} \mid do(\mv{d},\mv{t}) ]} \\
&= P_{\mathcal{M}}[v_{n_{j,i}} \mid  v_{n_{j,1}},\ldots,v_{n_{j,i-1}},do(\mv{d}, \mv{t}) ].
\end{align*}
\end{proof}

\end{document}